\documentclass[reqno,11pt]{amsart}
\usepackage{cases}
\usepackage{mathrsfs}
\usepackage[T1]{fontenc}
\usepackage{mathrsfs}
\usepackage{amsmath,latexsym,amssymb,amsfonts,amsbsy, amsthm}
\usepackage[usenames]{color}
\usepackage{xspace,colortbl}
\usepackage{epsfig}
\usepackage{amsmath,amsfonts,amsthm,amssymb,amscd}
\input amssym.def
\input amssym.tex

\newcommand{\Z }{\mathbb{Z} }

\newcommand{\LC}{\left(}
\newcommand{\RC}{\right)}

\newtheorem{theorem}{Theorem}[section]
\newtheorem{remark}[theorem]{Remark}
\newtheorem{lemma}[theorem]{Lemma}

\newtheorem{proposition}[theorem]{Proposition}

\newcommand{\ben}{\begin{eqnarray}}
\newcommand{\een}{\end{eqnarray}}
\newcommand{\beno}{\begin{eqnarray*}}
\newcommand{\eeno}{\end{eqnarray*}}

\def\FD#1{{\rm D}_{#1}}

\newdimen\eqjot \eqjot = 1\jot
\def\openupeq{\openup \the\eqjot}
\def\addtab#1={#1\;&=}
\def\addtabe#1=#2={#1=#2\;&=}

\def\qaeq#1#2{{\def\\{&}\vcenter{\openupeq\halign{$\displaystyle
   ##\hfil$&&\hskip#1pt$\displaystyle##\hfil$\cr #2\cr}}}}
 
\def\qxeq{\qaeq{30}} 
\def\qeq{\qaeq{20}}

\def\ezeq#1#2#3{{\def\\{\cr#1}\vcenter{\openupeq \halign{$\displaystyle
   \hfil##$&$\displaystyle##\hfil$&&\hskip#2pt$\displaystyle##\hfil
        $\cr#1#3\cr}}}}
\def\eaeq{\ezeq\addtab}

\def\eeq{\eaeq{20}}  

\def\caeq#1#2{{\def\\{\cr}\vcenter{\openupeq \halign
    {$\hfil\displaystyle ##\hfil$&&$\hfil\hskip#1pt\displaystyle ##\hfil
         $\cr#2\cr}}}}
\def\ceq{\caeq{20}}

\begin{document}

\title [Liouville correspondences]
{Liouville correspondences between \\multi-component integrable hierarchies}

\author{Jing Kang}
\address{Jing Kang\newline
Center for Nonlinear Studies and School of Mathematics, Northwest University, Xi'an 710069, P.R. China}
\email{jingkang@nwu.edu.cn}
\author{Xiaochuan Liu}
\address{Xiaochuan Liu\newline
School of Mathematics and Statistics, Xi'an Jiaotong University, Xi'an, Shaanxi 710049, P.R. China}
\email{liuxiaochuan@mail.xjtu.edu.cn}
\author{Peter J. Olver}
\address{Peter J. Olver\newline
School of Mathematics, University of Minnesota, Minneapolis, MN 55455, USA}
\email{olver@umn.edu}
\author{Changzheng Qu}
\address{Changzheng Qu\newline
Center for Nonlinear Studies and Department of Mathematics, Ningbo University, Ningbo 315211, P.R. China}
\email{quchangzheng@nbu.edu.cn}

\begin{abstract}
In this paper, we establish Liouville correspondences for the integrable two-component Camassa-Holm hierarchy, the two-component Novikov (Geng-Xue)  hierarchy, and the two-component dual dispersive water wave hierarchy by means of the related Liouville transformations. This extends previous results on the scalar Camassa-Holm and KdV hierarchies, and the Novikov and Sawada-Kotera hierarchies to the multi-component case.
\end{abstract}

\maketitle \numberwithin{equation}{section}




\section{Introduction}

This paper is concerned with Liouville correspondences for integrable hierarchies of multi-component systems possessing nonlinear dispersion, extending the previous work on scalar hierarchies \cite{kloq1, kloq2, len, mck}. Specifically, we shall  establish the Liouville correspondences among several multi-component integrable hierarchies, generated respectively by the two-component Camassa-Holm system, the two-component Novikov system, also known as the Geng-Xue system, as well as the two-component dual dispersive water wave system.

A basic idea for investigating the integrability of a new system is to establish its relationship with a known integrable system through some kind of transformation, including B\"{a}cklund transformations, Miura transformations, gauge transformations, Darboux transformations, hodograph transformations, Liouville transformations, etc. Application of an appropriate transformation enables one to derive solutions and analyze the integrability properties for the system under consideration through adaptation of known solutions and integrable structures. Among these vital features of integrability, the spatial isospectral problem in the Lax-pair formulation \cite{lax} plays a dominant role in the construction of solitons using inverse scattering transform, as well as the long-time behavior of solutions by virtue of the Riemann-Hilbert approach. Usually, the transition from one isospectral problem to another through a change of variables can be identified as a form of {\it Liouville transformation}; see also \cite{mil} and \cite{folv} for this terminology. It is then expected that such a {\it Liouville correspondence} can be used to establish an inherent correspondence between the associated integrability properties, including symmetries, conserved quantities, soliton solutions, Hamiltonian structures, etc.



In recent years, the intense interest in integrable systems of Camassa-Holm type has resulted from their novel properties, including nonlinear dispersion, (usually) supporting non-smooth soliton structures, such as peakons, cuspons, compactons, etc., and their ability to model wave-breaking phenomena. Previous investigations have established a variety of Liouville correspondences between integrable hierarchies of Camassa-Holm type and certain classical integrable hierarchies.  The most well-studied example is the Camassa-Holm (CH) equation
\ben\label{ch}
m_t+2u_xm+um_x=0, \qquad m=u-u_{xx},
\een
that has a quadratic nonlinearity \cite{ch, chh, cl, ff}.  In \cite{len} and \cite{mck}, the Liouville correspondence between the entire CH hierarchy induced by \eqref{ch} and the usual Korteweg--de Vries (KdV) hierarchy was established. Moreover, this provides a correspondence between the Hamiltonian functionals of the two hierarchies \cite{len}. Furthermore,  \cite{len1}, the Liouville transformation also relates their smooth traveling wave solutions.

 The modified Camassa-Holm (mCH) equation \cite{or}
\ben\label{mch}
m_t+\LC (u^2-u_x^2)\,m\RC_x=0, \qquad m=u-u_{xx},
\een
is a prototypical integrable model with cubic nonlinearity, which presents several novel properties, as described, for instance, in \cite{clqz, gloq, kloq1, llq2, loqz, mat}. The Liouville correspondence between the integrable mCH and modified Korteweg-de Vries (mKdV)  hierarchies, including the explicit relationships between their equations and Hamiltonian functionals, was established in \cite{kloq1}. In contrast to the CH-KdV situation, the analysis in \cite{kloq1} is based on the interrelationship between the respective recursion operators and the conservative structure of the mCH  hierarchy. Furthermore, with the respective Liouville correspondences between the CH-KdV hierarchies and the mCH-mKdV hierarchies in hand, a novel transformation mapping the mCH equation \eqref{mch} to the CH equation \eqref{ch} was constructed \cite{kloq1}.

It is worth noting that all of the equations in the CH, mCH, KdV, and mKdV hierarchies are  of bi-Hamiltonian form. Moreover, the two Hamiltonian operators for the CH and mCH integrable hierarchies can be constructed from those of the KdV and mKdV hierarchies, respectively, using the method of {\it tri-Hamiltonian duality}, established in \cite{for, fuc, or}. This approach is based on the observation that most standard integrable equations which possess a bi-Hamiltonian structure, actually admit a compatible trio of Hamiltonian structures through an adapted scaling argument. Recombinations of the members of the compatible Hamiltonian triple will generate different types of bi-Hamiltonian integrable systems, which admit a {\it dual} relationship. The tri-Hamiltonian duality relationships aid us to establish corresponding Liouville correspondences in the CH-KdV and mCH-mKdV cases as shown in \cite{kloq1} and \cite{len}.

For the  integrable hierarchies possessing generalized bi-Hamiltonian structures, i.e., compatible pairs of Dirac structures \cite{dorf}, or which do not admit tri-Hamiltonian duality, establishing the relevant Liouville correspondences is more challenging. Nevertheless, such Liouville correspondences are to be expected whenever Liouville transformations between the associated isospectral problems are provided.
In this case, the Novikov and Degasperis-Procesi (DP) \cite{dhh, dp} integrable hierarchies are two representative examples.
The Novikov integrable equation with cubic nonlinearity \cite{hw2, nov}
\ben\label{nov}
m_t=u^2\,m_x+3uu_xm, \qquad  m=u-u_{xx},
\een
is associated with a third-order isospectral problem. In \cite{hw2}, Hone and Wang employed a Liouville transformation to convert its isospectral problem into that of the Sawada-Kotera (SK)  equation \cite{cdg, sk}
\ben\label{sk}
Q_\tau+Q_{yyyyy}+5Q Q_{yyy}+5Q_yQ_{yy}+5Q^2Q_y=0.
\een
Even though the Novikov hierarchy is bi-Hamiltonian \cite{hw2}, its Hamiltonian operators do not support the tri-Hamiltonian duality construction, especially one that relates to the Hamiltonian operators of the SK equation. Moreover, the SK equation \eqref{sk} exhibits the generalized bi-Hamiltonian formulation with the corresponding  hierarchy generated by the recursion operator, which is the composition of symplectic and implectic operators that fail to satisfy the conditions of non-degeneracy or invertibility \cite{fo}. Therefore, producing a Liouville correspondence for the Novikov integrable hierarchy requires a more delicate analysis. In \cite{kloq2}, using the Liouville transformation and several operator decomposition identities, we were able to establish a Liouville correspondence between the Novikov and SK integrable hierarchies. In a similar manner, using a certain Liouville transformation proposed in \cite{dhh} and \cite{hw1}, a similar correspondence between the DP and Kaup-Kupershmidt \cite{kup}  hierarchies was also established in \cite{kloq2}.

Let us now turn our attention to  multi-component integrable systems of Camassa-Holm type.  One important example is the well-studied two-component CH (2CH) system \cite{clz}
\ben\label{2CH}
\begin{cases}
&\!m_t+2u_xm+u m_x+\rho\rho_x=0, \qquad m=u-u_{xx},\\
&\!\rho_t+(\rho u)_x=0,
\end{cases}
\een
which arises as an integrable shallow water model \cite{ci}. The system \eqref{2CH} is of significant interest, since it exhibits nonlinear interactions between the free surface and the horizontal velocity components, and can model the phenomenon of wave breaking; see, for example, \cite{eghkt, eg, ekl, go, gl, hi1, hnw, mat1, qsy, ss, xqz}. The 2CH system \eqref{2CH} is completely integrable and arises from the compatibility condition of the Lax-pair formulation \cite{ci}
\ben\label{lax-2ch}
\begin{aligned}
&\; \Psi_{xx} +\left(-\frac{1}{4}-\lambda m(t,x)+\lambda^2\rho^2(t,x)\right)\Psi=0, \\
&\; \quad \Psi_t=\left(\frac{1}{2\lambda}-u(t,x)\right)\Psi_x+\frac 12 u_x(t,x)\Psi.
\end{aligned}
\een
A particular Liouville transformation proposed in \cite{clz} will convert \eqref{lax-2ch} into
\ben\label{lax-2ach-1}
\begin{aligned}
\Phi_{yy}+\LC Q(\tau,y)+\lambda P(\tau,y)+\lambda^2\RC\Phi&=0,\\
\Phi_{\tau}-\frac{1}{2 \lambda}\rho(t,x)\,  \Phi_y+\frac{1}{4 \lambda}\rho_y(t,x)\Phi&=0,
\end{aligned}
\een
which is the Lax-pair formulation of the following integrable system
\ben
\label{2ach-1}
\ceq{P_{\tau}(\tau,y)=\rho_y, \qquad Q_{\tau}(\tau,y)=\frac 12 \rho\, P_y(\tau,y)+\rho_yP(\tau,y), \\
\rho_{yyy}+2\rho_yQ(\tau,y)+2(\rho\, Q(\tau,y))_y=0.
}
\end{eqnarray}
In \cite{clz},  taking into account the structure of spectrum in its Lax-pair formulation \eqref{lax-2ach-1}, the system \eqref{2ach-1} is recognized to be the first negative flow of the AKNS hierarchy \cite{akns}.

Moreover, the 2CH system \eqref{2CH} is bi-Hamiltonian --- and so the compatible Hamiltonian operators recursively generate the entire 2CH integrable hierarchy, with \eqref{2CH} forming the second flow in the positive direction. Moreover, the bi-Hamiltonian structure can be derived from that of the Ito system \cite{ito} using the method of tri-Hamiltonian duality \cite{or}. On the other hand, even though the 2CH and Ito integrable systems are in tri-Hamiltonian dual relationship, unlike the CH-KdV and mCH-mKdV cases, the Liouville correspondence between these two hierarchies is unexpected, because the transformation between the corresponding isospectral problems is not evident. Nevertheless, one expects to be able to establish a Liouville correspondence between the 2CH hierarchy and a second integrable hierarchy involving integrable system \eqref{2ach-1} as a particular member in the negative direction.

To achieve this goal, we need to overcome several new difficulties. First of all, the integrable structures, including the recursion operator and Hamiltonian operators, for the   hierarchy that includes system \eqref{2ach-1} as a negative flow is unclear. Chen {\rm et al} \cite{clz} also did not clarify the required integrability information. On the other hand, as in the scalar case, the verification of the Liouville correspondence relies on an analysis of the underlying operators, which in the multi-component case, are of matrix form and hence a more careful calculation of the nonlinear interplay among the various components is required. In the present paper, we elucidate the entire integrable hierarchy, which we call  associated with the system \eqref{2ach-1}, which forms the first negative flow in what we will refer to as the {\it associated two-component Camassa-Holm} (A2CH)  hierarchy.  We further demonstrate its bi-Hamiltonian structure and establish a Liouville correspondence between the 2CH and A2CH hierarchies. Furthermore, we find that the second positive flow of the A2CH integrable hierarchy is the following integrable system
\beno
Q_{\tau}=-\frac{1}{2}P_{yyy}-2\,QP_y-Q_yP,\qquad 
P_{\tau}=2\,Q_y-3PP_y,
\eeno
which belongs to the integrable family studied in \cite{il} and can be recognized as an integrable system of Kaup-Boussinesq type describing the motion of shallow water waves \cite{kau}.

The Novikov equation \eqref{nov} has the following two-component integrable generalization
\ben
\label{GX}
\eeq{m_t+3v u_xm+uv m_x=0,&m=u-u_{xx},\\
n_t+3uv_x n+uv n_x=0,&n=v-v_{xx},}
\een
which was introduced by Geng and Xue \cite{gx}, and so is referred to be the {\it Geng-Xue} (GX) system; see \cite{ls} and references therein. As a prototypical multi-component integrable system with cubic nonlinearity, the GX system \eqref{GX} admits special peakon solutions and has recently attracted much attention \cite{lc, ll-13, ln, ls, ls1}. In \cite{lc}, it was shown that there exists a certain Liouville transformation
converting the Lax-pair of the GX system \eqref{GX} into the Lax-pair of the following integrable system
\ben
\begin{cases}\label{aGX-1}
&Q_\tau=\frac{3}{2}(q_y+p_y)-(q-p) P,\qquad P_\tau=\frac{3}{2}(q-p),\\
&p_{yy}+2p_yP+pP_y+pP^2-pQ+1=0,\\
&q_{yy}-2q_yP-qP_y+qP^2-qQ+1=0,
\end{cases}
\een
where $q=v\,m^{2/3}n^{-1/3}$ and $p=u\,m^{-1/3}n^{2/3}$. The system \eqref{aGX-1} is bi-Hamiltonian, whose structure is derived in \cite{lc}. We shall elucidate the entire {\it associated Geng-Xue} (AGX) integrable hierarchy in which \eqref{aGX-1} is the first negative flow.  We also establish a Liouville correspondence between the integrable GX hierarchy generated by \eqref{GX} and the AGX hierarchy.

Finally, we consider the dual dispersive water wave (dDWW) integrable system
\ben\label{ddww}
\eeq{\rho_t=\LC (\rho+v)\,u\RC_x,&\rho=v-v_x,\\
\gamma_t=\LC \gamma \,u+2v\RC_x,&\gamma=u+u_x,}
\een
which was recently derived, \cite{kloq3}, from the dispersive water wave integrable system introduced by Kuperschimdt \cite{kup85}, using tri-Hamiltonian duality. The dDWW system \eqref{ddww} possesses a bi-Hamiltonian formulation and admits a variety of non-smooth soliton solutions \cite{kloq3}. We find a transformation to establish a Liouville correspondence between the dDWW integrable hierarchy generated by \eqref{ddww} with an associated dDWW integrable hierarchy. We provide the bi-Hamiltonian characterization for such associated dDWW integrable hierarchy and explore the explicit relationship between their flows and Hamiltonian functionals.

This section is concluded by outlining the rest of the paper. In Section 2, we first present a Liouville transformation relating the isospectral problems of the 2CH and A2CH integrable hierarchies. Next, we exploit the Liouville transformation to establish the one-to-one correspondence between the flows of the 2CH and A2CH hierarchies. Furthermore, we establish the relationship between the Hamiltonian functionals appearing in the 2CH and A2CH hierarchies. In Sections 3 and 4, we similarly investigate the Liouville correspondences for the GX and dDWW integrable hierarchies, respectively.

\section{The Liouville correspondence for \break the two-component Camassa-Holm hierarchy}

\subsection{The Liouville transformation for the isospectral problem of the 2CH system}

In this subsection, we shall present the explicit expression of the Liouville transformation for the isospectral problem of the 2CH system
\begin{equation}\label{2ch-s}
\begin{cases}
&m_t+u_x m+(u m)_x+\rho \rho_x=0, \quad m=u-u_{xx},\\
&\rho_t+\left(\rho u\right)_x=0,
\end{cases}
\end{equation}
which follows as the compatibility condition of the Lax-pair formulation
\begin{equation}\label{iso-2ch}
\Psi_{xx} +\left(-\frac{1}{4}-\lambda m+\lambda^2\rho^2\right) \Psi=0, \qquad \Psi_t=\left(\frac{1}{2\lambda}-u\right)\Psi_x+\frac{u_x}{2}\Psi,
\end{equation}
with the spectral parameter $\lambda$. It was proved in \cite{clz} that  the reciprocal transformation
\begin{equation}\label{tran1-2ch}
\mathrm{d} y=\rho\, \mathrm{d}x-\rho u\,\mathrm{d}t, \quad \mathrm{d} \tau=\mathrm{d} t,
\end{equation}
converts the isospectral problem (\ref{iso-2ch}) into
\begin{equation}\label{iso-akns}
\Phi_{yy}+(Q+\lambda P+\lambda^2)\Phi=0, \quad \Phi_{\tau}-\frac{1}{2 \lambda}\rho \,\Phi_y+\frac{1}{4 \lambda}\rho_y\Phi=0,
\end{equation}
with
\begin{equation}\label{Qm1}
\Phi=\sqrt{\rho}\,\Psi, \quad Q=-\frac{1}{4}\rho^{-2}+\frac{3}{4}\rho^{-4}\rho_x^2-\frac{1}{2}\rho^{-3}\rho_{xx},\quad P=-\frac{m}{\rho^2}.
\end{equation}
The linear equations \eqref{iso-akns} have the form of a Lax pair, and the resulting compatibility condition $\Phi_{yy\tau}=\Phi_{\tau yy}$ gives rise to the following integrable system
\begin{equation}\label{akns--1}
P_\tau=\rho_y, \qquad Q_\tau=\frac{1}{2}\rho P_y+\rho_yP, \quad \rho_{yyy}+2\rho_yQ+2(\rho\,Q)_y=0.
\end{equation}
Therefore, the anticipated Liouville correspondence between the 2CH system \eqref{2ch-s} and integrable system \eqref{akns--1} is provided by the reciprocal transformation \eqref{tran1-2ch} and the change of dependent variables \eqref{Qm1}. On the other hand, due to the spectral structure in \eqref{iso-akns}, the system \eqref{akns--1} can be viewed as the first negative flow of some specific integrable hierarchy that obeys the isospectral problem \eqref{iso-akns}, that will be referred to as the {\it associated 2CH} (A2CH)  hierarchy in this paper.

\subsection{The Liouville correspondence between the 2CH and A2CH hierarchies}

Motivated by these results, we are led to generalize the Liouville correspondence between the systems \eqref{2ch-s} and \eqref{akns--1}, to their respective entire integrable hierarchies. More precisely, we propose the following Liouville transformation
\begin{eqnarray}
\begin{aligned}\label{liou-2ch}
\tau&=t, \qquad y=\int^x\!\rho(t, \xi)\,\mathrm{d}\xi, \qquad P(\tau, y)=-m(t, x)\,\rho(t, x)^{-2},\\
Q(\tau, y)&=-\frac{1}{4}\rho(t, x)^{-2}+\frac{3}{4}\rho(t, x)^{-4}\rho_x^2(t, x)-\frac{1}{2}\rho(t, x)^{-3}\rho_{xx}(t, x).
\end{aligned}
\end{eqnarray}

First, the 2CH system \eqref{2ch-s} can be expressed in the bi-Hamiltonian form \cite{or}
\begin{eqnarray}
\begin{aligned}\label{bi-2ch}
\begin{pmatrix}
              m\\
             \rho
              \end{pmatrix}_t=\mathcal{K}\delta \mathcal{H}_1(m,\,\rho)=\mathcal{J}\delta \mathcal{H}_2(m,\,\rho), \quad \delta \mathcal{H}_n(m, \rho)=\left(\frac{\delta \mathcal{H}_n}{\delta m},\;\frac{\delta \mathcal{H}_n}{\delta \rho}\right)^T,\quad n=1,\,2,
\end{aligned}
\end{eqnarray}
with compatible Hamiltonian operators
\begin{eqnarray}
\begin{aligned}\label{haop-2ch}
\mathcal{K}=\begin{pmatrix}
              m \partial_x+\partial_x m& \rho\partial_x\\
              \partial_x \rho& 0\\
              \end{pmatrix}, \qquad
\mathcal{J}=\begin{pmatrix}
         \partial_x-\partial_x^3 & 0\\
             0& \partial_x\\
              \end{pmatrix}.
\end{aligned}
\end{eqnarray}
The associated Hamiltonian functionals are
\begin{equation*}
\qxeq{\mathcal{H}_1(m,\,\rho)=-\frac{1}{2} \int  (u^2+u_x^2+\rho^2)\;\mathrm{d}x , \quad  \mathcal{H}_2(m,\,\rho)=-\frac{1}{2}\int u (u^2+u_x^2+\rho^2)\;\mathrm{d}x.}
\end{equation*}
According to Magri's theorem \cite{mag}, the Hamiltonian pair induces the hierarchy
\begin{equation}\label{hie-2ch}
\begin{pmatrix}
              m\\
             \rho
              \end{pmatrix}_t=\mathbf{K}_n=\mathcal{K}\delta \mathcal{H}_{n-1}(m,\,\rho)=\mathcal{J}\delta \mathcal{H}_n(m,\,\rho), \quad \delta \mathcal{H}_n(m, \rho)=\left(\frac{\delta \mathcal{H}_n}{\delta m},\;\frac{\delta \mathcal{H}_n}{\delta \rho}\right)^T,\quad n\in \mathbb{Z} ,
\end{equation}
of commutative bi-Hamiltonian systems, based on the corresponding Hamiltonian functionals $\mathcal{H}_{n}=\mathcal{H}_{n}(m, \rho)$. The members in the hierarchy \eqref{hie-2ch} are obtained by successively applying  the recursion operator $\mathcal{R}=\mathcal{K}\,\mathcal{J}^{-1}$ to a seed symmetry \cite{olv1, olv2}. The positive flows of \eqref{hie-2ch} begin with the seed system
\begin{eqnarray*}
\begin{aligned}
\begin{pmatrix}
             m\\
           \rho
              \end{pmatrix}_t=\mathbf{K}_1=-\begin{pmatrix}
           m\\
           \rho
              \end{pmatrix}_x,
\end{aligned}
\end{eqnarray*}
and the 2CH system \eqref{bi-2ch} is the second member. Observe that the Hamiltonian operator $\mathcal{K}$ admits a Casimir functional
\begin{equation}\label{2ch-cas}
\mathcal{H}_C(m, \rho)=\int \frac{m}{\rho}\;\mathrm{d}x \qquad \mathrm{with\,\, variational\,\, derivative}\qquad \delta \mathcal{H}_C (m, \rho)= \begin{pmatrix}
             \rho^{-1}\\
           -m\rho^{-2}\end{pmatrix},
\end{equation}
which leads to an associated \emph{Casimir system}
\begin{eqnarray*}
\begin{aligned}
\begin{pmatrix}
              m\\
              \rho
              \end{pmatrix}_t=\mathbf{K}_{-1}=\mathcal{J}\delta \mathcal{H}_{C},
              \end{aligned}
\end{eqnarray*}
which serves as the first negative flow for the hierarchy \eqref{hie-2ch} and has the explicit form
\begin{equation}\label{2ch--1}
m_t=(\partial_x-\partial_x^3)\left(\displaystyle\frac{1}{\rho}\right),\qquad \rho_t= -\left(\displaystyle\frac{m}{\rho^2}\right)_x,\qquad m=u-u_{xx}.
\end{equation}
Applying the inverse recursion operator $\mathcal{R}^{-1}=\mathcal{J}\,\mathcal{K}^{-1}$ successively to \eqref{2ch--1} produces the members in the negative direction of \eqref{hie-2ch}, having the form
\begin{eqnarray}
\begin{aligned}\label{2ch--n}
\begin{pmatrix}
              m\\
              \rho
              \end{pmatrix}_t=\mathbf{K}_{-n}=(\mathcal{J}\,\mathcal{K}^{-1})^{n-1}\,\mathcal{J}\begin{pmatrix}
              \rho^{-1}\\
             -m\rho^{-2}
              \end{pmatrix},\qquad n=1,\,2,\ldots.
              \end{aligned}
\end{eqnarray}


Furthermore, we will demonstrate, in Lemma \ref{l2.2} below, that the A2CH integrable hierarchy involving system \eqref{akns--1} is actually generated by the following recursion operator
\begin{eqnarray}
\begin{aligned}\label{reop-akns}
\overline{\mathcal{R}}=\frac12\,\begin{pmatrix}
              0 & \partial_y^2+4Q+2Q_y\partial_y^{-1}\\
              -4 & 4P+2P_y\partial_y^{-1}\\
              \end{pmatrix}.
\end{aligned}
\end{eqnarray}
Successively applying $\overline{\mathcal{R}}$ to the usual seed symmetry $\overline{\mathbf{K}}_1=\left(-Q_y,\,-P_y\right)^T$ produces the positive flows of the A2CH integrable hierarchy:
\begin{eqnarray}
\begin{aligned}\label{akns-n+1}
\begin{pmatrix}
              Q\\
              P
              \end{pmatrix}_\tau=\overline{\mathbf{K}}_{n}=\overline{\mathcal{R}}^{n-1}\,\overline{\mathbf{K}}_1,\quad n=1,\,2,\ldots.
              \end{aligned}
\end{eqnarray}
On the other hand, in the negative direction, since the trivial symmetry $\overline{\mathbf{K}}_0=\left(0,\,0\right)^T$ satisfies $\overline{\mathcal{R}}\,\overline{\mathbf{K}}_{0}=\overline{\mathbf{K}}_1$, the negative flow at stage $n\in \mathbb{Z}^+$ takes the form
\begin{eqnarray}
\begin{aligned}\label{akns--n}
\overline{\mathcal{R}}^n\begin{pmatrix}
              Q\\
              P
              \end{pmatrix}_\tau=\overline{\mathbf{K}}_0,\quad n=1,\,2,\ldots.
              \end{aligned}
\end{eqnarray}
In particular, the first negative flow,  $n=1$, in \eqref{akns--n} takes the  explicit form
\begin{equation}\label{akns--1'}
\left(\frac{1}{2}\partial_y^2+2Q+Q_y\partial_y^{-1}\right)P_\tau=0,\quad Q_\tau=\left(P+\frac{1}{2}P_y \partial_y^{-1}\right)P_\tau.
\end{equation}
More precisely, the system \eqref{akns--1} arising from the compatibility condition of the Lax-pair formulation \eqref{iso-akns} is a reduction of the first negative flow \eqref{akns--1'}. Furthermore, the second positive flow, for $n=2$, of the A2CH hierarchy in \eqref{akns-n+1} is
\begin{eqnarray}
\begin{aligned}\label{akns-2}
\begin{pmatrix}
              Q\\
              P
              \end{pmatrix}_\tau=\overline{\mathbf{K}}_2=\overline{\mathcal{R}}\,\overline{\mathbf{K}}_{1}=\overline{\mathcal{R}}\begin{pmatrix}
              -Q_y\\
              -P_y
              \end{pmatrix}=\begin{pmatrix}
              -\frac{1}{2}P_{yyy}-2QP_y-Q_yP\\
              2Q_y-3PP_y
              \end{pmatrix},
              \end{aligned}
\end{eqnarray}
which can be obtained by the $y$ component of the Lax-pair formulation \eqref{iso-akns} together with
\beno
\Phi_\tau+(2\lambda+P)\Phi_y-\frac{1}{2}P_y\Phi=0.
\eeno
In \cite{il}, Ivanov and Lyons introduced the following general Lax-pair formulation
\begin{eqnarray*}
\begin{aligned}\label{iso-il-y}
\Phi_{yy}&=\left(-\lambda^2+\lambda \,u(\tau, y)+\frac{\kappa}{2}u^2(\tau, y)+v(\tau, y)\right)\Phi,\\
\Phi_\tau&=-\left(\lambda+\frac{1}{2}u(\tau, y)\right)\Phi_y+\frac{1}{4}u_y(\tau, y)\Phi,
  \end{aligned}
\end{eqnarray*}
where $\kappa$ is an arbitary constant, leading to the integrable system
\begin{equation*}
 \begin{cases}
  &u_\tau+v_y+(\frac{3}{2}+\kappa)u u_y=0,\\
  &v_\tau-\frac{1}{4}u_{yyy}+(u v)_y-(\frac{1}{2}+\kappa)u v_y-\kappa(\frac{1}{2}+\kappa)u^2 u_y=0.
 \end{cases}
\end{equation*}
If $\kappa=0$, then the resulting system is exactly \eqref{akns-2} up to a change of variables $u=P$, $v=-Q$ and $\tau=2 t$. Moreover, a change of dependent variable
\beno
Q(\tau, y)=N(\tau, y)+\frac{1}{4}P^2(\tau, y),
\eeno
converts  \eqref{akns-2} into
\begin{equation*}
N_\tau=-\frac{1}{2}P_{yyy}-2(PN)_y,\quad
 P_\tau=2N_y-2PP_y,
\end{equation*}
which has the form of a Kaup-Boussinesq system considered in \cite{kau}.

Hereafter, we denote, for a positive integer $n$,  the $n$-th equation in the positive and negative directions of the 2CH hierarchy \eqref{hie-2ch} by (2CH)$_n$ and (2CH)$_{-n}$, respectively, while the $n$-th positive and negative flows of the A2CH hierarchy are denoted by (A2CH)$_{n}$ and (A2CH)$_{-n}$, respectively. With these notations, we are in a position to state the main result on the Liouville correspondence between the 2CH and A2CH hierarchies as follows.

\begin{theorem}\label{t1}
Under the Liouville transformation \eqref{liou-2ch}, for each integer $n$, the {\rm(}2CH{\rm)}$_{n+1}$ equation is mapped into the {\rm(}A2CH{\rm)}$_{-n}$ equation.
\end{theorem}

The proof of this theorem relies on the following two preliminary lemmas.


\begin{lemma}\label{l2.1}
Let $\left(m(t, x),\,\rho(t, x)\right)$ and $\left(Q(\tau, y),\,P(\tau, y)\right)$ be related by the transformation \eqref{liou-2ch}. Then the following operator identities hold:
\begin{eqnarray}
\label{op1}
\eeq{\rho^{-\frac{3}{2}}\,\left(\frac{1}{4}-\partial_x^2\right)\,\rho^{-\frac{1}{2}}=-(Q+\partial_y^2),\\
 \rho^{-2}\,\left(\partial_x-\partial_x^3\right)\,\rho^{-1}=-\LC \partial_y^3+2Q\partial_y+2\partial_y Q\RC,\\
 \rho^{-2}\,(\partial_x m+m \partial_x)\,\rho^{-1}=-(P\partial_y+\partial_y P).}
\end{eqnarray}
\end{lemma}

\begin{proof}
First of all, in view of the transformation \eqref{liou-2ch}, one has
\begin{equation}\label{eq-xy}
\partial_x=\rho\,\partial_y.\end{equation}
 It follows that
\begin{equation*}
\partial_x^2\,\rho^{-\frac{1}{2}}=\frac{3}{4}\rho^{-\frac{5}{2}}\,\rho_x^2-\frac{1}{2}\rho^{-\frac{3}{2}}\rho_{xx}+\rho^{\frac{3}{2}}\,\partial_y^2.
\end{equation*}
We thus arrive at
\begin{eqnarray*}
\rho^{-\frac{3}{2}}\,\left(\frac{1}{4}-\partial_x^2\right)\,\rho^{-\frac{1}{2}}=\rho^{-\frac{3}{2}}\left(\frac{1}{4}\rho^{-\frac{1}{2}}-\frac{3}{4}\rho^{-\frac{5}{2}}\,\rho_x^2+\frac{1}{2}\rho^{-\frac{3}{2}}\rho_{xx}-\rho^{\frac{3}{2}}\,\partial_y^2\right)=-\LC Q+\partial_y^2\RC,
\end{eqnarray*}
proving the first identity in \eqref{op1}.
Next, from \eqref{eq-xy},  one has
\beno
\eeq{\partial_x \rho^{-1}=-\rho^{-1}\rho_y+\partial_y,\\
\partial_x^3\rho^{-1}=-\rho^{-1}\rho_y^3+2\rho_y\rho_{yy}-\rho\rho_{yyy}+(\rho_y^2-2\rho\rho_{yy})\partial_y+\rho^2\partial_y^3.}
\eeno
Hence,
\begin{eqnarray*}
\begin{aligned}
\rho^{-2}\,\left(\partial_x-\partial_x^3\right)&\,\rho^{-1}\\
=&-\rho^{-3}\rho_y+\rho^{-3}\rho_y^3-2\rho^{-2}\rho_y\rho_{yy}+\rho^{-1}\rho_{yyy}+(\rho^{-2}-\rho^{-2}\rho_y^2+2\rho^{-1}\rho_{yy})\partial_y-\partial_y^3,
\end{aligned}
\end{eqnarray*}
and then the second identity in \eqref{op1} follows.
Finally, a direct computation shows that
 \begin{eqnarray*}
\begin{aligned}
\rho^{-2}\partial_x\,m\rho^{-1}+m\rho^{-2}\partial_x\,\rho^{-1}&=-\rho^{-1}\partial_yP\rho- P\rho\partial_y\rho^{-1}\\
 &=-\rho^{-1}\big(\,(P\rho)_y+P\rho\partial_y\,\big)+P\big(\,\rho^{-1}\rho_y-\partial_y\,\big)=-(P\partial_y+\partial_y P),
\end{aligned}
\end{eqnarray*}
verifying the third identity in \eqref{op1}.
\end{proof}

\begin{lemma}\label{l2.2}
Let $\mathcal{K}$ and $\mathcal{J}$ be the two compatible Hamiltonian operators given in \eqref{haop-2ch} for the 2CH integrable hierarchy. Then, for each positive integer $n$, \begin{equation}\label{reop-eq-2ch}
\mathbf{A}^{-1}\left( \mathcal{J}\,\mathcal{K}^{-1}\right)^n \mathbf{A}=\overline{\mathcal{R}}^n,
\end{equation}
through the transformations \eqref{liou-2ch}, where the operator $\overline{\mathcal{R}}$ is defined by \eqref{reop-akns} and
\begin{eqnarray}
\begin{aligned}\label{reop-eq-A}
\mathbf{A}=\begin{pmatrix}
              -2\rho^2& 0\\
              0&  \rho\\
              \end{pmatrix}.
\end{aligned}
\end{eqnarray}
Moreover, $\overline{\mathcal{R}}$ is not only a hereditary operator itself, but also the recursion operator for the  A2CH hierarchy with positive flows given by \eqref{akns-n+1} and negative flows by \eqref{akns--n}.
\end{lemma}

Before proving Lemma \ref{l2.2}, we give a brief remark on the notion of hereditary operators in the two-component setting. Let $\mathcal{A}$ denotes the space of differential functions, depending only on the indicated dependent variables and their spatial derivatives, while $\mathcal{A}^n$  denotes the corresponding space of $n$-component differential functions. An operator, say $\overline{\mathcal{R}}$ as in \eqref{reop-akns} for instance, is called a {\it hereditary operator} if and only if it satisfies the condition
\begin{equation}\label{here-op}
\FD{\overline{\mathcal{R}}}[\overline{\mathcal{R}}\,\mathbf{f}]\,\mathbf{g}-\FD{\overline{\mathcal{R}}}[\overline{\mathcal{R}}\,\mathbf{g}]\,\mathbf{f}=\overline{\mathcal{R}}\left(\FD{\overline{\mathcal{R}}}[\mathbf{f}]\mathbf{g}-\FD{\overline{\mathcal{R}}}[\mathbf{g}]\mathbf{f}\right), \quad \textrm{for all} \quad\mathbf{f},\,\mathbf{g}\in \mathcal{A}^2.
\end{equation}
The following proposition, proved in \cite{ff, olv2}, describes how a hereditary operator serves as the recursion operator of an integrable hierarchy.

\begin{proposition}\label{p2.1}
Assume that the hereditary operator $\overline{\mathcal{R}}$ and the system
\begin{equation}\label{sys}
\mathbf{Q}_\tau=\mathbf{G}_1,\quad \mathbf{Q}=(Q(\tau, y),\,P(\tau, y))^T,\quad \mathbf{G}_1=\left( G_1^1(\mathbf{Q}),\, G_1^2(\mathbf{Q})\right)^T\in \mathcal{A}^2,
\end{equation}
satisfy the condition
\begin{equation}\label{re-op}
\overline{\mathcal{R}}_\tau\Big|_{\mathbf{Q}_\tau=\mathbf{G}_1}=\Big[ \FD{\mathbf{G}_1},\;\overline{\mathcal{R}}\Big],
\end{equation}
where $\overline{\mathcal{R}}_\tau$ is the time derivative of $\overline{\mathcal{R}}$ and $\FD{\mathbf{G_1}}$ is the Fr\'echet derivative of $\mathbf{G_1}$, so that $\overline{\mathcal{R}}$ is the recursion operator for \eqref{sys}. Then, $\overline{\mathcal{R}}$ is also the recursion operator for each flow in the assocaited hierarchy
\begin{equation*}
\mathbf{Q}_\tau=\mathbf{G}_{n}=\overline{\mathcal{R}}^{n-1}\mathbf{G}_1,\qquad n\in \mathbb{Z}.
\end{equation*}
\end{proposition}

\begin{proof} [\bf{Proof of Lemma \ref{l2.2}}]
We prove \eqref{reop-eq-2ch} by induction. For the case $n=1$, in view of the forms of $\mathcal{K}$ and $\mathcal{J}$, we obtain
\begin{eqnarray*}
\begin{aligned}
\mathbf{A}^{-1}\mathcal{J}\mathcal{K}^{-1}\mathbf{A}=\frac{1}{2}\,\begin{pmatrix}
            0 & -\rho^{-2}(\partial_x-\partial_x^3)\rho^{-1}\partial_x^{-1}\rho\\
              -4 & -2\rho^{-2}(m\partial_x+\partial_x m)\rho^{-1}\partial_x^{-1}\rho\\
              \end{pmatrix},
\end{aligned}
\end{eqnarray*}
which together with the second and third equations in \eqref{op1} leads to \eqref{reop-eq-2ch}  for $n=1$. Next, assume that  \eqref{reop-eq-2ch} holds for $n=k$ with some $k\in \mathbb{Z}^+$. Then for $n=k+1$, one has
\begin{equation*}
\overline{\mathcal{R}}^{k+1}=\overline{\mathcal{R}}^k \overline{\mathcal{R}}=\mathbf{A}^{-1}(\mathcal{J}\mathcal{K}^{-1})^k\mathbf{A} \mathbf{A}^{-1}\mathcal{J}\mathcal{K}^{-1}\mathbf{A}=\mathbf{A}^{-1}(\mathcal{J}\mathcal{K}^{-1})^{k+1}\mathbf{A},
\end{equation*}
which establishes the induction step and thus verifies  \eqref{reop-eq-2ch} holds for any integer $n\geq 1$.

According to Proposition \ref{p2.1}, to prove that the operator $\overline{\mathcal{R}}$ serves as a recursion operator for the entire A2CH hierarchy, we first verify that $\overline{\mathcal{R}}$ is a recursion operator for the seed system
\begin{eqnarray}
\begin{aligned}\label{akns-1}
\begin{pmatrix}
              Q\\
              P
              \end{pmatrix}_\tau=\overline{\mathbf{K}}_1=-\begin{pmatrix}
              Q\\
              P
              \end{pmatrix}_y.
              \end{aligned}
\end{eqnarray}
In order to prove \eqref{re-op}, on the one hand, the relevant Fr\'echet derivative is
\begin{eqnarray*}
\begin{aligned}
\FD{\overline{\mathbf{K}}_1}=-\begin{pmatrix}
      \partial_y&0\\
              0& \partial_y\\
              \end{pmatrix}.
\end{aligned}
\end{eqnarray*}
On the other hand, with respect to \eqref{akns-1}, we have
\begin{eqnarray*}
\begin{aligned}\label{rtau}
\overline{\mathcal{R}}_\tau=-\begin{pmatrix}
       0&Q_{yy}\partial_y^{-1}+2Q_y\\
              0& P_{yy}\partial_y^{-1}+2P_y\\
              \end{pmatrix}.
\end{aligned}
\end{eqnarray*}
Furthermore, the commutator in \eqref{re-op} satisfies
\begin{eqnarray*}
\begin{aligned}
\Big[ & \FD{\overline{\mathbf{K}}_1},\;\overline{\mathcal{R}}\Big]=\FD{\overline{\mathbf{K}}_1} \cdot \overline{\mathcal{R}}-\overline{\mathcal{R}} \cdot \FD{\overline{\mathbf{K}}_1}\\
&=\frac{1}{2}\begin{pmatrix}
       0 & -\partial_y\LC \partial_y^2+2Q_y\partial_y^{-1}+4Q\RC\\
              4\partial_y & -2\partial_y(2P+P_y\partial_y^{-1})\\
              \end{pmatrix}-\frac{1}{2}\begin{pmatrix}
       0 & -\LC \partial_y^2+2Q_y\partial_y^{-1}+4Q\RC\partial_y\\
              4\partial_y & -2(2P+P_y\partial_y^{-1})\partial_y\\
              \end{pmatrix}=\overline{\mathcal{R}}_{\tau},
\end{aligned}
\end{eqnarray*}
which verifies \eqref{re-op} and demonstrates that $\overline{\mathcal{R}}$ is a recursion operator for \eqref{akns-1}.

Finally, a direct and tedious calculation shows that $\overline{\mathcal{R}}$ satisfies the hereditary property \eqref{here-op}. This, together with the fact that $\overline{\mathcal{R}}$ is a recursion operator for the seed system \eqref{akns-1} suffices to prove that it is a recursion operator for each flow in the A2CH integrable hierarchy, which completes the proof of this lemma.
\end{proof}

\begin{proof} [\bf{Proof of Theorem \ref{t1}}]
First, we compute the $t$-derivatives of the functions $Q(\tau, y)$ and $P(\tau, y)$ appearing in the Liouville transformation \eqref{liou-2ch}. More precisely, \begin{equation}\label{t-deri-1}
Q_t=Q_\tau+Q_y\int^x\!\rho_t(t, \xi)\,\mathrm{d}\xi=Q_\tau+Q_y\partial_y^{-1}\rho^{-1}\rho_t.
\end{equation}
On the other hand, note that $Q$ given in the transformation \eqref{liou-2ch} can be rewritten in the following compact form
\begin{equation*}
Q=-\rho^{-\frac{3}{2}}\,\LC \frac{1}{4}-\partial_x^2\RC\,\rho^{-\frac{1}{2}},
\end{equation*}
which implies
\begin{equation}\label{t-deri-2}
Q_t=-\LC \rho^{-\frac{3}{2}}\,\LC \frac{1}{4}-\partial_x^2\RC\,\rho^{-\frac{1}{2}}\RC_t=-\LC 2Q+\frac{1}{2}\partial_y^2\RC\,\rho^{-1} \rho_t,
\end{equation}
by using the operator identity \eqref{op1}. Thus, combining with \eqref{t-deri-1} and \eqref{t-deri-2}, we obtain
\begin{equation*}
Q_\tau=-\LC \frac{1}{2}\partial_y^3+2Q\partial_y+Q_y\RC\partial_y^{-1}\rho^{-1}\rho_t.
\end{equation*}
Similarly,
\begin{equation*}
P_\tau=-\rho^{-2}m_t-(P\partial_y+\partial_y P)\partial_y^{-1}\rho^{-1}\rho_t,
\end{equation*}
and hence\begin{eqnarray}
\begin{aligned}\label{t1QP}
\begin{pmatrix}
              Q\\
              P
              \end{pmatrix}_\tau=-\overline{\mathcal{R}}\,\mathbf{A}^{-1}\begin{pmatrix}
              m\\
              \rho
              \end{pmatrix}_t,
              \end{aligned}
\end{eqnarray}
where $\mathbf{A}$ is defined in \eqref{reop-eq-A}.

Next, consider the (2CH)$_{-n}$ system \eqref{2ch--n} for $n\geq 1$. The second identity in \eqref{op1} and the transformation formulae \eqref{liou-2ch}  imply that the first negative flow \eqref{2ch--1} of the 2CH hierarchy satisfies
\begin{eqnarray*}
\begin{aligned}
\begin{pmatrix}
              m\\
              \rho
              \end{pmatrix}_t=\mathbf{K}_{-1}=\mathbf{A}\begin{pmatrix}
              Q\\
              P
              \end{pmatrix}_y,
              \end{aligned}
\end{eqnarray*}
and hence the $n$-th negative flow \eqref{2ch--n} can be written as
\begin{eqnarray}
\begin{aligned}\label{2ch--nQP}
\begin{pmatrix}
              m\\
              \rho
              \end{pmatrix}_t=\mathbf{K}_{-n}=\LC\mathcal{J} \mathcal{K}^{-1}\RC^{n-1}\,\mathbf{K}_{-1}
=\;\LC\mathcal{J} \mathcal{K}^{-1}\RC^{n-1}\,\mathbf{A}\begin{pmatrix}
              Q\\
              P
              \end{pmatrix}_y,\quad n=1,\,2,\ldots.
\end{aligned}
\end{eqnarray}
Plugging \eqref{2ch--nQP} into \eqref{t1QP} and using the formula \eqref{reop-eq-2ch}, we arrive at
\begin{eqnarray*}
\begin{aligned}
\begin{pmatrix}
              Q\\
              P
              \end{pmatrix}_\tau=-\overline{\mathcal{R}}\mathbf{A}^{-1}\LC\mathcal{J} \mathcal{K}^{-1}\RC^{n-1}\mathbf{A}\begin{pmatrix}
              Q\\
              P
              \end{pmatrix}_y
=-\;\mathbf{A}^{-1}\LC\mathcal{J} \mathcal{K}^{-1}\RC^{n}\mathbf{A}\begin{pmatrix}
              Q\\
              P
              \end{pmatrix}_y=\overline{\mathcal{R}}^{n}\overline{\mathbf{K}}_1=\overline{\mathbf{K}}_{n+1}.
              \end{aligned}
\end{eqnarray*}
Therefore, for each $n\geq 1$, if $\big(m(t, x),\,\rho(t, x)\big)$ is the solution of the (2CH)$_{-n}$ system \eqref{2ch--n}, then the corresponding $\big(Q(\tau, y),\,P(\tau, y)\big)$ solves the (A2CH)$_{n+1}$ system \eqref{akns-n+1}.

Moreover, for  $0 \leq n\ \in \Z$, substituting the (2CH)$_{n+1}$ system in the positive direction
\begin{eqnarray}
\begin{aligned}\label{2ch-n+1}
\begin{pmatrix}
              m\\
              \rho
              \end{pmatrix}_t=\mathbf{K}_{n+1}=\;-\LC\mathcal{K}\mathcal{J}^{-1}\RC^{n}\begin{pmatrix}
              m\\
              \rho
              \end{pmatrix}_x
              \quad n=0,\,1,\ldots,
\end{aligned}
\end{eqnarray}
into \eqref{t1QP} yields
\begin{eqnarray*}
\begin{aligned}
\begin{pmatrix}
              Q\\
              P
              \end{pmatrix}_\tau=\overline{\mathcal{R}}\mathbf{A}^{-1}\LC \mathcal{K}\mathcal{J}^{-1}\RC^n\,\begin{pmatrix}
              m\\
             \rho
              \end{pmatrix}_x.
              \end{aligned}
\end{eqnarray*}
Then, we let the operator $\overline{\mathcal{R}}^n$ to act on the both sides of the above system, and use the formula \eqref{reop-eq-2ch} again, to deduce that
\begin{eqnarray*}
\begin{aligned}
\overline{\mathcal{R}}^n\begin{pmatrix}
              Q\\
              P
              \end{pmatrix}_\tau=\overline{\mathcal{R}}^{n+1}\mathbf{A}^{-1}\LC \mathcal{K}\mathcal{J}^{-1}\RC^n\,\begin{pmatrix}
              m_x\\
             \rho_x
              \end{pmatrix}=\mathbf{A}^{-1}\mathcal{J}\mathcal{K}^{-1}\begin{pmatrix}
              m_x\\
             \rho_x
              \end{pmatrix}=\begin{pmatrix}
              0\\
             0
              \end{pmatrix},
              \end{aligned}
\end{eqnarray*}
where the operation $\mathcal{K}^{-1} \big(m_x, \rho_x\big)^T=\big(1, 0\big)^T$ is used. We conclude that, for each $n\geq 0$,  if $\big(m(t, x),\,\rho(t, x)\big)$ is the solution of the (2CH)$_{n+1}$ system \eqref{2ch-n+1}, then the corresponding $\big(Q(\tau, y),\,P(\tau, y)\big)$ solves the (A2CH)$_{-n}$ system \eqref{akns--n}.

Finally, if $\big(m(t, x),\,\rho(t, x)\big)$ solves the (2CH)$_0$ system, then the corresponding $\big(Q(\tau, y),\,P(\tau, y)\big)$ satisfies
\begin{eqnarray*}
\begin{aligned}
\begin{pmatrix}
              Q\\
              P
              \end{pmatrix}_\tau=-\overline{\mathcal{R}}\mathbf{A}^{-1}\begin{pmatrix}
              0\\
             0
              \end{pmatrix}=\begin{pmatrix}
              -Q_y\\
             -P_y
              \end{pmatrix},
              \end{aligned}
\end{eqnarray*}
which implies that $\big(Q(\tau, y),\,P(\tau, y)\big)$ is the solution of the (A2CH)$_{1}$ system.
\end{proof}

\subsection{The correspondence between the Hamiltonian functionals of the 2CH and A2CH hierarchies}

In this subsection, we investigate the correspondence between the Hamiltonian functionals appearing in the 2CH and A2CH hierarchies. We first consider the effect of the Liouville transformation on the Hamiltonian operators.
To do this, we need some relevant results found in \cite{olv88,olv2}. In \cite{olv88}, a general transformation formula for scalar Hamiltonian operators under a change of variables was provided, which can be directly generalized to the multi-component case.


Consider an n-component system of equations
\begin{eqnarray}\label{sys-m}
\mathbf{m}_t=\mathbf{K}(\mathbf{m}),\qquad \mathbf{m}=\big(m_1(t, x)\,,\ldots,m_n(t, x)\big)^T,
\end{eqnarray}
where $\mathbf{K}(\mathbf{m})=\big(K_1(\mathbf{m})\,,\ldots,\,K_n(\mathbf{m})\big)^T\,\in \mathcal{A}^n$ is an $n$-component  differential function depending on $\mathbf{m}$ and their $x$-derivatives up to a given order. Assume that system \eqref{sys-m} and another n-component system involving the dependent variable $\mathbf{Q}$:
\begin{eqnarray}\label{sys-Q}
\mathbf{Q}_\tau=\overline{\mathbf{K}}(\mathbf{Q}),\qquad \mathbf{Q}=\big(Q_1(\tau, y)\,,\ldots,Q_n(\tau, y)\big)^T,
\end{eqnarray}
where $\overline{\mathbf{K}}(\mathbf{Q})=\big(\overline{K}_1(\mathbf{Q})\,,\ldots,\,\overline{K}_n(\mathbf{Q})\big)^T\,\in \mathcal{A}^n$, are related by the following coordinate transformation
\begin{eqnarray}
\begin{aligned}\label{tran}
y=\int^x\!\Lambda(\mathbf{m})\,\mathrm{d}\xi\,\equiv I(x,\,\mathbf{m}), \qquad \tau=t, \qquad
&Q_i=F_i(\mathbf{m}), \qquad i=1, \ldots, n,
\end{aligned}
\end{eqnarray}
where $\Lambda(\mathbf{m})$ and $F_i(\mathbf{m})\in \mathcal{A}$.

\begin{theorem}\label{t2}
Let  $\mathcal{D}(\mathbf{m})$ be a Hamiltonian operator of system \eqref{sys-m}. Suppose $\mathbf{m}(t, x)$ and $\mathbf{Q}(\tau, y)$ satisfy \eqref{sys-m} and \eqref{sys-Q} respetively, and are related by the coordinate transformation \eqref{tran}. Then the corresponding Hamiltonian operator $\overline{\mathcal{D}}(\mathbf{Q})$ of system \eqref{sys-Q} takes the form
\begin{equation}\label{rela-haop-n}
\overline{\mathcal{D}}(\mathbf{Q})=\Lambda(\mathbf{m})^{-1}\mathbf{T}\,\mathcal{D}(\mathbf{m})\mathbf{T}^\ast,
\end{equation}
where $\mathbf{T}$ is the $n\times n$ matrix operator with entries
\begin{equation}\label{T}
\big(\mathbf{T}\big)_{i, j}=\Lambda(\mathbf{m})\FD{F_i,m_j}-D_xQ_i\,\FD{I,m_j}, \qquad  i,\,j=1,\ldots, n,
\end{equation}
and
$\mathbf{T}^\ast$ is its (formal) $L^2$ adjoint.  Here, $D_x$ denotes the total derivative with respect to $x$, while $\FD{I,m_j}$ and $\FD{F_i,m_j}$ are the  Fr\'echet derivatives of $I$ and $F_i$ with respect to $m_j$.
\end{theorem}

We adapt the proof of Theorem \ref{t2} in the scalar case  given in \cite{olv88}, using the following lemma \cite{olv2}.

\begin{lemma}\label{l2.3}
Let $\mathcal{H}(\mathbf{m})$ and $\bar{\mathcal{H}}(\mathbf{Q})$ be two functionals related by the change of variables \eqref{tran}. Then their corresponding variational derivatives $\delta\mathcal{H}(\mathbf{m})=\big(\delta \mathcal{H}/\delta m_1\,,\ldots,\,\delta \mathcal{H}/\delta m_n\big)^T$ and  $\delta\bar{\mathcal{H}}(\mathbf{Q})=\big(\delta \bar{\mathcal{H}}/\delta Q_1\,,\ldots,\,\delta \bar{\mathcal{H}}/\delta Q_n\big)^T$ satisfy
\begin{equation}\label{vdn}
\delta\mathcal{H}(\mathbf{m})=\mathbf{T}^\ast\,\delta \bar{\mathcal{H}}(\mathbf{Q}),
\end{equation}
where $\mathbf{T}^\ast$ is the formal adjoint  of the operator $\mathbf{T}$ given in \eqref{T}.
\end{lemma}

\begin{proof} [\bf{Proof of Theorem \ref{t2}}]
If \begin{equation}\label{sys-m-1}\mathbf{m}_t=\mathbf{K}(\mathbf{m})=\mathcal{D}\,\delta
\mathcal{H}(\mathbf{m})
\end{equation}
is a Hamiltonian system in the $\mathbf{m}(t, x)$ variable, then the corresponding evolution equation in the $\mathbf{Q}(\tau, y)$ variable will also be Hamiltonian
\begin{equation}\label{sys-Q-1}
\mathbf{Q}_\tau=\overline{\mathbf{K}}(\mathbf{Q})=\overline{\mathcal{D}}\,\delta
\bar{\mathcal{H}}(\mathbf{Q}),
\end{equation}
with the Hamiltonian operator $\overline{\mathcal{D}}$ given in \eqref{rela-haop-n}. To prove this,
we first write the transformation \eqref{tran} in implicit form
\begin{equation}\label{batr-vec}
\mathbf{B}(\mathbf{Q},  \mathbf{m})=\big(B_1(\mathbf{Q},  \mathbf{m})\,,\ldots,\,B_n(\mathbf{Q},  \mathbf{m})\big)^T=\mathbf{0},
\end{equation}
where $B_i(\mathbf{Q},  \mathbf{m})=Q_i-F_i( \mathbf{m})$ for $ i=1,\ldots,n.$

Next, taking the $t$-derivative of each equation in the system \eqref{batr-vec} yields
\begin{eqnarray*}
\begin{aligned}
\sum_{k=1}^{n}\,\left( \FD{B_i,m_k}\,m_{k, t}+\FD{B_i,Q_k}\,Q_{k, \tau}\right)=0,\quad i=1\,,\ldots,\,n,
\end{aligned}
\end{eqnarray*}
where $\FD{B_i,m_k}$ and $\FD{B_i,Q_k}$ are the Fr\'echet derivatives of $B_i$ with respect to $m_k$ and $Q_k$, respectively.
For convenience, the above expression can be written in vectorial form:
\begin{eqnarray*}
\begin{aligned}\label{batr-vec-Qm}
\mathbf{B}_{\mathbf{m}}\,\mathbf{m}_t+\mathbf{B}_{\mathbf{Q}}\,\mathbf{Q}_\tau=\mathbf{0}.
\end{aligned}
\end{eqnarray*}
From the transformation \eqref{batr-vec}, $\mathbf{B}_{\mathbf{m}}=-\Lambda(\mathbf{m})^{-1}\mathbf{T}$. This immediately gives rise to
\begin{eqnarray*}
\begin{aligned}\label{batr-eq-kg'}
\mathbf{Q}_\tau=\Lambda(\mathbf{m})^{-1}\mathbf{T}\,\mathbf{m}_t.
\end{aligned}
\end{eqnarray*}
Finally, assembing \eqref{sys-m-1} and \eqref{sys-Q-1}, together with formula \eqref{vdn} verifies \eqref{rela-haop-n}, completing the proof of the theorem.
\end{proof}

A direct application of the above theorem is to derive a compatible pair of Hamiltonian operators for the A2CH integrable hierarchy, which is generated by the recursion operator $\overline{\mathcal{R}}$ in \eqref{reop-akns}, from the given Hamiltonian operators $\mathcal{K}$ and $\mathcal{J}$ \eqref{haop-2ch} of the 2CH integrable hierarchy. Indeed, due to transformation \eqref{liou-2ch}, the function $\Lambda(\cdot)$ given in \eqref{tran} takes the form $\Lambda(m, \rho)=\rho$. Additionally, the operator defined in \eqref{T} is given by
\begin{eqnarray*}
\begin{aligned}
\mathbf{T}=-\frac{1}{2}\begin{pmatrix}
              0 & \rho(\partial_y^2+4Q+2Q_y\partial_y^{-1})\rho^{-1}\\
            2\rho^{-1} & \rho(4P+2P_y\partial_y^{-1})\rho^{-1}\\
              \end{pmatrix}
\end{aligned}
\end{eqnarray*}
with adjoint
\begin{eqnarray}\label{T-adjoint}
\begin{aligned}
\mathbf{T}^\ast=-\frac{1}{2}\begin{pmatrix}
       0 & 2\rho^{-1}\\
            \partial_y^2+4Q-2\partial_y^{-1}Q_y & 4P-2\partial_y^{-1}P_y\\
              \end{pmatrix}.
\end{aligned}
\end{eqnarray}
We conclude that the resulting pair of Hamiltonian operators for the A2CH hierarchy are
\begin{eqnarray}
\begin{aligned}\label{haop-A2ch}
\overline{\mathcal{K}}&=\Lambda^{-1}\mathbf{T}\mathcal{J}\mathbf{T}^\ast\\
&=\frac{1}{4}\begin{pmatrix}
             \mathcal{L}\,\partial_y^{-1}\mathcal{L} & 2\mathcal{L}\,\partial_y^{-1}(P\partial_y+\partial_y P)\\
              2(P\partial_y+\partial_y P)\partial_y^{-1}\mathcal{L} & 4(P\partial_y+\partial_y P)\partial_y^{-1}(P\partial_y+\partial_y P)+2\mathcal{L}\\
              \end{pmatrix},\\
\overline{\mathcal{J}}&=\Lambda^{-1}\mathbf{T}\mathcal{K}\mathbf{T}^\ast=\frac{1}{2}\begin{pmatrix}
       0 & \mathcal{L}\\
            \mathcal{L} & 2(P\partial_y+\partial_y P)\\
              \end{pmatrix},
\end{aligned}
\end{eqnarray}
where $\mathcal{L}=\partial_y^3+2Q\partial_y+2\partial_y Q$, and hence the A2CH hierarchy can be written in bi-Hamiltonian form. In the positive direction,
\begin{equation*}
\begin{pmatrix}
             Q\\
             P
              \end{pmatrix}_\tau=\overline{\mathbf{K}}_n=\overline{\mathcal{K}}\delta \bar{\mathcal{H}}_{n-1}=\overline{\mathcal{J}}\delta \bar{\mathcal{H}}_n, \qquad \delta \bar{\mathcal{H}}_n=\left(\frac{\delta \bar{\mathcal{H}}_n}{\delta Q},\;\frac{\delta \bar{\mathcal{H}}_n}{\delta P}\right)^T,\qquad n\geq 1,
\end{equation*}
with, for example,
\begin{equation}\label{a2ch-fun-12}
\qxeq{\bar{\mathcal{H}}_{1}=- \int  P\;\mathrm{d}y, \\\bar{\mathcal{H}}_{2}=-\int  \left(\frac{1}{2}P^2+2Q\right)\;\mathrm{d}y}.
\end{equation}

\begin{remark}
It follows from \eqref{reop-akns} and \eqref{haop-A2ch} that
\begin{eqnarray*}
\begin{aligned}
\overline{\mathcal{R}}=\overline{\mathcal{K}}\,\overline{\mathcal{J}}^{-1}=\frac{1}{2}\begin{pmatrix}
              0 & \partial_y^2+2Q+Q_y\partial_y^{-1}\\
              -4 & 4P+2P_y\partial_y^{-1}\\
              \end{pmatrix}
              =\Lambda^{-1}\mathbf{T} \mathcal{J}\,\mathcal{K}^{-1}\mathbf{T}^{-1}\Lambda.
\end{aligned}
\end{eqnarray*}
This formula provides an alternative decomposition of the recursion operator $\overline{\mathcal{R}}$, differing from the one given in Lemma \ref{l2.2}.
\end{remark}

Now, with the two pairs of Hamiltonian operators $\{\mathcal{K}, \mathcal{L}\}$ and $\{\overline{\mathcal{K}}, \overline{\mathcal{L}}\}$ in hand, Magri's theorem enables us to  recursively construct  an infinite hierarchy of Hamiltonian functionals. In the 2CH setting, these are determined by the recursive formula
\begin{equation}\label{2chclhie}
\mathcal{K}\,\delta \mathcal{H}_{n-1}=\mathcal{J}\,\delta \mathcal{H}_n,\qquad n\in  \mathbb{Z},
\end{equation}
whereas for the A2CH hierarchy they satisfy
\begin{equation}\label{aknsclhie}
\overline{\mathcal{K}}\delta \bar{\mathcal{H}}_{n-1}=\overline{\mathcal{J}} \delta \bar{\mathcal{H}}_{n},\qquad n\in \mathbb{Z}.
\end{equation}
In order to establish the correspondence between the respective hierarchies of Hamiltonian functionals $\mathcal{H}_{n}$ and  $\bar{\mathcal{H}}_{n}$, we require a formula which clarifies the correspondence between their variational derivatives.

\begin{lemma}\label{l2.4}
Let $\{\mathcal{H}_n\}$ and $\{\bar{\mathcal{H}}_n\}$  be the hierarchies of Hamiltonian functionals determined by the recurrence formulas \eqref{2chclhie} and \eqref{aknsclhie}, respectively. Then their corresponding variational derivatives satisfy
\begin{equation}\label{eq-l2.4}
\delta \mathcal{H}_{-(n+1)}(m, \rho)=-\mathcal{K}^{-1}\mathbf{A}\overline{\mathcal{J}}\,\delta  \bar{\mathcal{H}}_n(Q, P), \qquad 0\neq n\in \mathbb{Z}.
\end{equation}
\end{lemma}

\begin{proof}
We first prove \eqref{eq-l2.4} for $n\geq 1$ by induction. Since in the 2CH setting, $\mathcal{H}_{-1}=\mathcal{H}_C(m, \rho)$, the Casimir functional given in \eqref{2ch-cas}, then $\delta \mathcal{H}_{-1}=\delta \mathcal{H}_C=\big(\rho^{-1},\,-m\rho^{-2}\big)^T$. Using \eqref{op1}  and \eqref{2chclhie}, we have
\begin{eqnarray*}
\delta \mathcal{H}_{-2}=\mathcal{K}^{-1}\mathcal{J}\delta \mathcal{H}_{-1}=\mathcal{K}^{-1}\mathbf{A}\LC Q_y, \,P_y\RC^{T}=-\mathcal{K}^{-1}\mathbf{A}\overline{\mathcal{J}}\delta \bar{\mathcal{H}}_1,
\end{eqnarray*}
verifying that \eqref{eq-l2.4} holds for $n=1$. Assume now  \eqref{eq-l2.4} holds for $n=k$ with some integer $k\geq 1$; in other words,
\begin{equation*}
\delta \mathcal{H}_{-(k+1)}=-\mathcal{K}^{-1}\mathbf{A}\overline{\mathcal{J}}\,\delta  \bar{\mathcal{H}}_k.
\end{equation*}
Then, for $n=k+1$,
\begin{eqnarray*}
\begin{aligned}
\delta \mathcal{H}_{-(k+2)}&=\mathcal{K}^{-1}\mathcal{J}\delta \mathcal{H}_{-(k+1)}=-\mathcal{K}^{-1}\mathcal{J}\mathcal{K}^{-1}\mathbf{A}\overline{\mathcal{J}}\,\delta  \bar{\mathcal{H}}_k\\
&=-\mathcal{K}^{-1}\mathcal{J}\mathcal{K}^{-1}\mathbf{A}\overline{\mathcal{J}}\,\overline{\mathcal{K}}^{-1}\overline{\mathcal{J}}\,\delta  \bar{\mathcal{H}}_{k+1}=-\mathcal{K}^{-1}\mathbf{A}\overline{\mathcal{J}}\,\delta \bar{\mathcal{H}}_{k+1},
\end{aligned}
\end{eqnarray*}
where we have made use of the relation \eqref{reop-eq-2ch} with $n=1$. This proves \eqref{eq-l2.4} for each integer $n\geq 1$.

Next, for the case of  $n=-1$, note first that \eqref{eq-l2.4} is equivalent to
\begin{equation}\label{eq2-l2.4-1}
\delta \bar{\mathcal{H}}_{-1}=-\overline{\mathcal{J}}^{-1}\mathbf{A}^{-1}\mathcal{K}\,\delta \mathcal{H}_0.
\end{equation}
Since $\bar{\mathcal{H}}_{-1}$ is a Casimir functional of the Hamiltonian operator $\overline{\mathcal{K}}$, in order to prove \eqref{eq2-l2.4-1}, it suffices to verify that the function $h=-\overline{\mathcal{J}}^{-1}\mathbf{A}^{-1}\mathcal{K}\,\delta \mathcal{H}_0$ satisfies $\overline{\mathcal{K}}\, h=\mathbf{0}$. Indeed, using \eqref{reop-eq-2ch} with $n=1$ once again, one has
\begin{equation*}\label{eq2-l2.4}
\overline{\mathcal{K}}\, h=\mathbf{A}^{-1}\mathcal{J}\mathcal{K}^{-1}\mathbf{A}\mathbf{A}^{-1}\mathcal{K}\,\delta \mathcal{H}_0=\mathbf{A}^{-1}\mathcal{J}\,\delta \mathcal{H}_0=\mathbf{0},
\end{equation*}
where we take advantage of the property that $\delta \mathcal{H}_0$ is a constant vector.
Next, assuming that \eqref{eq-l2.4} holds for $n=-k \leq -1$, we infer that
\begin{eqnarray*}
\begin{aligned}
\delta \mathcal{H}_{-(-k-1+1)}&=\delta \mathcal{H}_k=\mathcal{J}^{-1}\mathcal{K}\delta \mathcal{H}_{k-1}=-\mathcal{J}^{-1}\mathbf{A}\overline{\mathcal{J}}\delta \bar{\mathcal{H}}_{-k}\\
&=-\mathcal{J}^{-1}\mathbf{A}\overline{\mathcal{K}}\,\delta \bar{\mathcal{H}}_{-k-1}=-\mathcal{K}^{-1}\mathbf{A}\overline{\mathcal{J}}\,\delta \bar{\mathcal{H}}_{-k-1},
\end{aligned}
\end{eqnarray*}
which, by induction, establishes \eqref{eq-l2.4} for each $n\leq -1$ and thus proves the lemma in general.
\end {proof}

Finally, given that $\big(m(t, x), \rho(t, x)\big)$ and $\big(Q(\tau, y), P(\tau, y)\big)$ are related by the transformation \eqref{liou-2ch}, we define the functional
\begin{equation*}
\mathcal{G}_n(Q, P)\equiv \mathcal{H}_{n}(m, \rho).
\end{equation*}
A direct application of Lemma \ref{l2.3} leads to
\begin{equation*}\label{t3-1}
\delta \mathcal{H}_{n}(m, \rho)=\mathbf{T}^\ast \delta \mathcal{G}_n(Q,  P),
\end{equation*}
with the operator $\mathbf{T}^\ast$ defined in \eqref{T-adjoint}. On the other hand, by Lemma \ref{l2.4}, we have
\begin{equation*}
\delta \mathcal{H}_{n}(m, \rho)=-\mathcal{K}^{-1}\mathbf{A}\overline{\mathcal{J}}\,\delta \bar{\mathcal{H}}_{-(n+1)}(Q, P),
\end{equation*}
which, together with \eqref{2chclhie} and relation \eqref{reop-eq-2ch}, yields
\begin{equation*}\label{t3-2}
\delta \mathcal{H}_{n}(m, \rho)= -\mathcal{J}^{-1}\mathbf{A}\overline{\mathcal{J}}\,\delta \bar{\mathcal{H}}_{-n}(Q, P).
\end{equation*}
Furthermore, a direct calculation shows that the operator $\mathbf{T}^\ast$ admits the following decomposition
\begin{equation*}
\mathcal{J}^{-1}\mathbf{A}\overline{\mathcal{J}}=-\mathbf{T}^\ast,
\end{equation*}
which implies
\begin{equation*}
\mathcal{H}_n(m, \rho)=\mathcal{G}_n(Q, P)=\bar{\mathcal{H}}_{-n}(Q, P).
\end{equation*}
The following theorem is thus proved.

\begin{theorem}\label{t3}
Under the Liouville transformation \eqref{liou-2ch}, for each nonzero integer $n$, the Hamiltonian functionals $\mathcal{H}_n(m, \rho)$ of the 2CH hierarchy given by \eqref{2chclhie} are related to the Hamiltonian functionals $\bar{\mathcal{H}}_{n}(Q, P)$ of the A2CH hierarchy given by \eqref{aknsclhie}, according to
\begin{equation*}\label{cl-rela-dp1}
\mathcal{H}_n(m, \rho)=\bar{\mathcal{H}}_{-n}(Q, P), \qquad 0\neq n\in \mathbb{Z}.
\end{equation*}
\end{theorem}

For instance, in the case of $n=-1$, using \eqref{liou-2ch}, one can check that
\begin{equation*}
\mathcal{H}_{-1}(m, \rho)= \int  \frac{m}{\rho}\;\mathrm{d}x=-\int P\;\mathrm{d}y=\bar{\mathcal{H}}_{1}(Q, P).
\end{equation*}
In the case of $n=-2$, \eqref{a2ch-fun-12} and \eqref{eq-l2.4} imply
\begin{equation*}
\delta \mathcal{H}_{-2}= \left(-\frac{m}{\rho^3},\,-2Q+\frac{3}{2}P^2\right)^T,
\end{equation*}
and hence
\begin{equation*}\label{cl-rela-dp2}
\mathcal{H}_{-2}(m, \rho)= -\int  \rho \left(\frac{1}{2}P^2+2Q\right)\;\mathrm{d}x=-\int  \left(\frac{1}{2}P^2+2Q\right)\;\mathrm{d}y=\bar{\mathcal{H}}_{2}(Q, P).
\end{equation*}
Both of the preceding two results are in accordance with Theorem \ref{t3}.

\section{The Liouville correspondence for the Geng-Xue hierarchy}

\subsection{The Liouville transformation for the isospectral problem of the Geng-Xue system}
The Lax-pair formulation for the GX system
\begin{equation}\label{gx}
 \begin{cases}
  &m_t+3vu_x m+uv m_x=0,\\
  &n_t+3uv_x n+uv n_x=0,\qquad m=u-u_{xx}, \quad n=v-v_{xx},
 \end{cases}
\end{equation}
takes the form \cite{gx}
\begin{equation}\label{iso-gx}
\eeq{
\mathbf{\Psi}_x
=\begin{pmatrix} 0\; &\lambda\, m\;& 1\\
             0\; &0\; &\lambda\, n\\
             1\; &0\; & 0
               \end{pmatrix}\mathbf{\Psi},\qquad \mathbf{\Psi}
=\begin{pmatrix} \psi_1\\
           \psi_2\\
             \psi_3
               \end{pmatrix}
,\\\mathbf{\Psi}_t
=\begin{pmatrix} -u_xv\; &\lambda^{-1} u_x-\lambda uv m\;& u_xv_x\\
             \lambda^{-1} v\; &-\lambda^{-2}+u_x v-uv_x\; &-\lambda uv n-\lambda^{-1} v_x\\
             -uv\; &\lambda^{-1} u\; &u v_x
               \end{pmatrix}\mathbf{\Psi}.}
\end{equation}

It was shown in \cite{lc} that the reciprocal transformation defined by
\begin{equation}\label{liou-gx-xy1}
\mathrm{d} y=\Delta^{\frac{1}{3}} \mathrm{d}x-\Delta^{\frac{1}{3}} uv\mathrm{d}t, \qquad \mathrm{d} \tau=\mathrm{d} t,\qquad \Delta=mn,
\end{equation}
converts the isospectral problem \eqref{iso-gx} into
\begin{equation}\label{iso-agx}
\eeq{\mathbf{\Phi}_y
=\begin{pmatrix} 0\; &\lambda\;& Q\\
             0\; &P\; &\lambda\\
             1\; &0\; & 0
               \end{pmatrix}\mathbf{\Phi},\qquad \mathbf{\Phi}
=\begin{pmatrix} \phi_1\\
           \phi_2\\
             \phi_3
               \end{pmatrix},
\\\mathbf{\Phi}_\tau
=\frac{1}{2}\begin{pmatrix}
        A\; & 2\lambda^{-1}(p_y+p\,P) & p+q\\
           2\lambda^{-1} q\; & A-2\lambda^{-2} & 2\lambda^{-1}(Pq-q_y)\\
            0\; & 2\lambda^{-1} p\; & A
               \end{pmatrix} \mathbf{\Phi},}
\end{equation}
respectively, where
\begin{equation}\label{liou-gx-PQ1}
Q=\frac{1}{6}\Delta^{-\frac{5}{3}}\Delta_{xx}-\frac{7}{36}\Delta^{-\frac{8}{3}}\Delta_x^2+\Delta^{-\frac{2}{3}},\quad P=\frac{1}{2}m^{-\frac{4}{3}}n^{\frac{2}{3}}\left(\frac{m}{n}\right)_x
\end{equation}
and
\begin{equation*}
A=q_yp-qp_y-2pqP,\quad q=vm^{\frac{2}{3}}n^{-\frac{1}{3}},\quad p=um^{-\frac{1}{3}}n^{\frac{2}{3}}.
\end{equation*}
The compatibility condition $\mathbf{\Phi}_{y\tau}=\mathbf{\Phi}_{\tau y}$ gives rise to the following integrable system
\begin{equation}\label{agx--1}
 \begin{cases}
  &Q_\tau=\frac{3}{2}(q_y+p_y)-(q-p) P,\quad P_\tau=\frac{3}{2}(q-p),\\
  &p_{yy}+2p_yP+pP_y+pP^2-pQ+1=0,\\
  &q_{yy}-2q_yP-qP_y+qP^2-qQ+1=0,
 \end{cases}
\end{equation}
which admits \eqref{iso-agx} as the corresponding Lax-pair formulation.

It turns out that there exists a Liouville correspondence between the GX system \eqref{gx} and system \eqref{agx--1}, in the sense that their respective isospectral problems \eqref{iso-gx} and \eqref{iso-agx} are related through the transformations \eqref{liou-gx-xy1}, \eqref{liou-gx-PQ1}. In addition, in view of the spectral structure of the time evolution component of \eqref{iso-agx}, the reduced system  \eqref{agx--1} can be viewed as a negative flow of an integrable hierarchy, namely the associated Geng-Xue (AGX) integrable hierarchy.

\subsection{The Liouville correspondence between the GX and AGX integrable hierarchies}

Consider the following transformation
\begin{eqnarray}
\begin{aligned}\label{liou-gx}
y&=\int^x\!\Delta^{\frac{1}{3}}(t,\,\xi)\,\mathrm{d}\xi, \quad \tau=t,\\
Q&=\Delta^{-\frac{2}{3}}+\frac{1}{6}\Delta^{-\frac{5}{3}}\Delta_{xx}-\frac{7}{36}\Delta^{-\frac{8}{3}}\Delta_x^2,\quad P=\frac{1}{2}m^{-\frac{4}{3}}n^{\frac{2}{3}}\left(\frac{m}{n}\right)_x,
\end{aligned}
\end{eqnarray}
where and throughout this section $\Delta=mn$.

The GX system \eqref{gx}  can be written in a bi-Hamiltonian form \cite{ll-13}
\begin{eqnarray*}
\begin{aligned}
\begin{pmatrix}
              m\\
            n
              \end{pmatrix}_t=\mathcal{K}\delta \mathcal{H}_1(m,\,n)=\mathcal{J}\delta \mathcal{H}_2(m,\,n),
\end{aligned}
\end{eqnarray*}
where the compatible Hamiltonian operators are
\begin{eqnarray}
\begin{aligned}\label{haop-gx}
&\mathcal{K}=\frac{3}{2}\begin{pmatrix}
             3 m^{\frac{1}{3}} \partial_x m^{\frac{2}{3}} \Omega^{-1} m^{\frac{2}{3}} \partial_x m^{\frac{1}{3}}+m\partial_x^{-1}m & 3 m^{\frac{1}{3}} \partial_x m^{\frac{2}{3}} \Omega^{-1} n^{\frac{2}{3}} \partial_x n^{\frac{1}{3}}-m\partial_x^{-1}n\\
             3 n^{\frac{1}{3}} \partial_x n^{\frac{2}{3}} \Omega^{-1} m^{\frac{2}{3}} \partial_x m^{\frac{1}{3}}-n\partial_x^{-1}m& 3 n^{\frac{1}{3}} \partial_x n^{\frac{2}{3}} \Omega^{-1} n^{\frac{2}{3}} \partial_x n^{\frac{1}{3}}+3 n\partial_x^{-1}n\\
              \end{pmatrix},\\
&
  \mathcal{J}=\begin{pmatrix}
                           0 & \partial_x^2-1\\
                            1-\partial_x^2 & 0\\
                       \end{pmatrix}, \qquad {\rm where} \quad \Omega=\partial_x^3-4\,\partial_x,
\end{aligned}
\end{eqnarray}
while
\begin{equation*}
\qxeq{\mathcal{H}_1(m,\,n)= \int  un\;\mathrm{d}x,  \quad  \mathcal{H}_2(m,\,n)=\int  (u_x v-uv_x)un\;\mathrm{d}x}
\end{equation*}
are the initial Hamiltonian functionals.
Based on the bi-Hamiltonian structure of the GX system, one may construct the full GX integrable hierarchy by applying the resulting hereditary recursion operator $\mathcal{R}=\mathcal{K}\,\mathcal{J}^{-1}$ to the particular seed system
\begin{eqnarray*}
\begin{aligned}
\begin{pmatrix}
             m\\
         n
              \end{pmatrix}_t=\mathbf{G}_1(m, n)=\begin{pmatrix}
           -m\\
           n
              \end{pmatrix}.
\end{aligned}
\end{eqnarray*}
Hence, the $l$-th member in the positive direction takes the form
\begin{equation}\label{gx-l}
\begin{pmatrix}
              m\\
             n
              \end{pmatrix}_t=\mathbf{G}_{l}(m, n)=\mathcal{R}^{l-1}\,\mathbf{G}_1(m, n),\quad l=1, 2, \ldots,
\end{equation}
and the GX system \eqref{gx} is exactly the second positive flow. In analogy with the 2CH hierarchy, the fact that the trivial symmetry $\mathbf{G}_{0}=\big( 0,\,0\big)^T$ satisfies $\mathcal{R}\mathbf{G}_{0}=\mathbf{G}_{1}$, implies that, in the opposite direction, the negative flow is generated by the Casimir system. Note that the Hamiltonian operator $\mathcal{K}$ admits the following Casimir functional
\begin{equation}\label{hacl--1-gx}
\qxeq{\mathcal{H}_C(m,\,n)=3 \int  \Delta^{\frac{1}{3}}\;\mathrm{d}x \quad \hbox{with variational derivative}\quad \delta \mathcal{H}_C=\big(m^{-\frac{2}{3}}n^{\frac{1}{3}},\, m^{\frac{1}{3}}n^{-\frac{2}{3}}\big)^T.}
\end{equation}
Therefore, the $l$-th negative flow of the GX hierarchy is
\begin{eqnarray}
\begin{aligned}\label{gx--l'}
\begin{pmatrix}
              m\\
              n
              \end{pmatrix}_t=\mathbf{G}_{-l}(m, n)=(\mathcal{J}\,\mathcal{K}^{-1})^{l-1}\,\mathcal{J}\delta \mathcal{H}_C,\qquad l=1,\,2,\ldots.
              \end{aligned}
\end{eqnarray}

Turning to the AGX integrable hierarchy, based on Theorem \ref{t2}, one can readily construct two Hamiltonian operators for the transformed system \eqref{agx--1} by applying the Liouville transformation \eqref{liou-gx} from the given Hamitonian pair $\mathcal{K}$ and $\mathcal{J}$ introduced in \eqref{haop-gx} for the GX system. Indeed, resulting Hamiltonian operators  admitted by \eqref{agx--1} were given in \cite{lc} as follows:
\begin{eqnarray}
\begin{aligned}\label{haop-kj-agx}
\overline{\mathcal{K}}=\mathbf{\Gamma}\begin{pmatrix}
            0\, & \Theta\\
            -\Theta^\ast& 0\\
              \end{pmatrix}\mathbf{\Gamma}^\ast \quad \hbox{and}\quad
              \overline{\mathcal{J}}= \frac{1}{2}\begin{pmatrix}
                                        \mathcal{E} & 0\,\\
                                          0\, & -3\partial_y\\
                                            \end{pmatrix},
\end{aligned}
\end{eqnarray}
where the matrix operator $\mathbf{\Gamma}$, and operators $\Theta$, $\mathcal{E}$ are defined by
\begin{eqnarray}
\begin{aligned}\label{haop-ga-agx}
 \mathbf{\Gamma}&=-\frac{1}{6}\begin{pmatrix}
        \mathcal{E}\partial_y^{-1} & \mathcal{E}\partial_y^{-1}\\
            (3\partial_y^2-2\partial_y P)\partial_y^{-1}&  -(3\partial_y^2+2\partial_y P)\partial_y^{-1}\\
              \end{pmatrix},
\\\Theta&=\partial_y^2+P\partial_y+\partial_y P+P^2-Q, \\  \mathcal{E}&=\partial_y^3-2Q\partial_y-2\partial_y Q.
\end{aligned}
\end{eqnarray}
Therefore, the AGX integrable hierarchy can be obtained by applying the resulting hereditary recursion operator $\overline{\mathcal{R}}=\overline{\mathcal{K}}\,\overline{\mathcal{J}}^{-1}$ to the seed system
\begin{equation*}
\begin{pmatrix}
             Q\\
             P
              \end{pmatrix}_\tau=\overline{\mathbf{G}}_1=-\begin{pmatrix}
             Q\\
             P
              \end{pmatrix}_y.
\end{equation*}
More precisely, the $l$-th member in the positive direction takes the form
\begin{equation}\label{agx-l+1}
\begin{pmatrix}
             Q\\
             P
              \end{pmatrix}_\tau=\overline{\mathbf{G}}_{l}=\overline{\mathcal{R}}^{l-1}\overline{\mathbf{G}}_1,\quad l=1, 2, \ldots.
\end{equation}
While, the $l$-th negative flow can be written as
\begin{equation}\label{agx--l}
\overline{\mathcal{R}}^l\begin{pmatrix}
             Q\\
             P
              \end{pmatrix}_\tau=\overline{\mathbf{G}}_0=\begin{pmatrix}
             0\\
             0
              \end{pmatrix}, \quad l=1, 2, \ldots.
\end{equation}

The scheme of the Liouville correspondence between the GX and AGX integrable hierarchies is described in the following Theorem 3.1. Adopting a similar notation to that in Theorem \ref{t1}, the $l$-th positive and negative flows of the GX and AGX hierarchies are denoted by  (GX)$_l$ and (GX)$_{-l}$, and by (AGX)$_{l}$ and (AGX)$_{-l}$, respectively.

\begin{theorem}\label{t3.1}
Under the Liouville transformation \eqref{liou-gx}, for each integer $l \geq 1$,

{\bf(i).} If $\big(m(t, x),\,n(t, x)\big)$ is a solution of the  (GX)$_l$ system \eqref{gx-l}, then the corresponding\break $\big(Q(\tau, y),\,P(\tau, y)\big)$ satisfies the (AGX)$_{-l}$ system \eqref{agx--l};

{\bf(ii).} If $\big(m(t, x),\,n(t, x)\big)$ is a solution of the (GX)$_{-l}$ system \eqref{gx--l'}, then the corresponding $\big(Q(\tau, y),\,P(\tau, y)\big)$ satisfies the (AGX)$_{l+1}$ system \eqref{agx-l+1}.
\end{theorem}

Two preliminary lemmas are required.
\begin{lemma}\label{l3.1}
Let $\big(m(t, x),\,n(t, x)\big)$ and $\big(Q(\tau, y),\,P(\tau, y)\big)$ be related through the transformation \eqref{liou-gx}. Then the following operator identities hold:
\begin{equation}
\label{gx-op}
\qeq{\Delta^{-\frac{1}{2}}\,\left(1-\partial_x^2\right)\,\Delta^{-\frac{1}{6}}=Q-\partial_y^2, \quad
\Delta^{-\frac{2}{3}}\,\Omega\,\Delta^{-\frac{1}{3}}=\mathcal{E}, \quad
 m^{-\frac{4}{3}}\,n^{\frac{2}{3}}\partial_x mn^{-1}=2P+\partial_y,
 }
\end{equation}
where $\Omega$ and $\mathcal{E}$ are given in \eqref{haop-gx} and \eqref{haop-ga-agx}, respectively.
\end{lemma}

\begin{proof}
{\bf(i).} In view of the transformation \eqref{liou-gx}, one has
\begin{equation}\label{l3.1-eq1}
\partial_x=\Delta^{\frac 13} \partial_y,\end{equation}
which implies that
\begin{equation*}
\partial_x^2\,\Delta^{-\frac{1}{6}}=\frac{7}{36}\Delta^{-\frac{13}{6}}\,\Delta_x^2-\frac{1}{6}\Delta^{-\frac{7}{6}}\Delta_{xx}-\frac{1}{3}\Delta^{-\frac{7}{6}}\Delta_x\,\partial_x+\Delta^{-\frac{1}{6}}\,\partial_x^2.
\end{equation*}
Substituting this into
$\Delta^{-\frac{1}{2}}\,\left(1-\partial_x^2\right)\,\Delta^{-\frac{1}{6}}$ yields the first identity in \eqref{gx-op}.

{\bf(ii).} From \eqref{l3.1-eq1}, one has
\begin{equation*}
\partial_x^2\,\Delta^{-\frac{1}{3}}=\frac{4}{9}\Delta^{-\frac{1}{3}}\,\Delta_x^2-\frac{1}{3}\Delta^{-\frac{4}{3}}\Delta_{xx}-\frac{1}{3}\Delta^{-1}\Delta_x\,\partial_y+\Delta^{\frac{1}{3}}\,\partial_y^2.
\end{equation*}
Then a straightforward computation shows that
\begin{eqnarray*}
\begin{aligned}
\Delta^{-\frac{2}{3}}\,&\Omega\,\Delta^{-\frac{1}{3}}=\Delta^{-\frac{2}{3}}\,\partial_x\,\left(\partial_x^2-4\right)\,\Delta^{-\frac{1}{3}}\\
=&\Delta^{-\frac{2}{3}}\,\partial_x\,\left(-4\Delta^{-\frac{1}{3}}+\frac{4}{9}\Delta^{-\frac{7}{3}}\Delta_x^2-\frac{1}{3}\Delta^{-\frac{4}{3}}\Delta_{xx}-\frac{1}{3}\Delta^{-1}\Delta_{x}\,\partial_y+\Delta^{\frac{1}{3}}\,\partial_y^2\right)\\
=&\frac{4}{3}\Delta^{-2}\Delta_{x}-\frac{28}{27}\Delta^{-4}\Delta_{x}^3+\frac{4}{3}\Delta^{-3}\Delta_{x}\Delta_{xx}-\frac{1}{3}\Delta^{-2}\Delta_{xxx}\\
&\quad \quad \quad +\left(-4\Delta^{-\frac{2}{3}}+\frac{7}{9}\Delta^{-\frac{8}{3}}\Delta_{x}^2-\frac{2}{3}\Delta^{-\frac{5}{3}}\Delta_{xx}\right)\,\partial_y+\,\partial_y^3.
\end{aligned}
\end{eqnarray*}
Using the formula for $Q$ in \eqref{liou-gx}, we deduce the second identity in \eqref{gx-op}.

{\bf(iii).} In view of \eqref{l3.1-eq1} again, one has
\begin{equation*}
m^{-\frac{4}{3}}\,n^{\frac{2}{3}}\partial_x mn^{-1}=\frac{n}{m}\,\partial_y\,\frac{m}{n}
=\frac{n}{m}\left(\frac{m}{n}\right)_y+\partial_y.
\end{equation*}
Hence, the third identity in \eqref{gx-op} follows from the formula for $P$ in \eqref{liou-gx}.
\end{proof}

\begin{lemma}\label{l3.2}
Define
\begin{eqnarray}\label{agx-e-b}
\begin{aligned}
\mathbf{E}=\begin{pmatrix}
              \partial_y& 0\\
              0&\partial_y\\
              \end{pmatrix}\qquad \hbox{and}\qquad \mathbf{B}=\begin{pmatrix}
              m& 0\\
              0&n\\
              \end{pmatrix}.
\end{aligned}
\end{eqnarray}
Then the operator identity
\begin{equation}\label{eq-l3.2}
\mathbf{B}^{-1}\left( \mathcal{J}\,\mathcal{K}^{-1}\right)^l \mathbf{B}=(-1)^l\,\mathbf{E}\,\left(\mathbf{\mathcal{V}}\,\mathbf{\mathcal{U}}\right)^l\,\mathbf{E}^{-1}
\end{equation}
holds for each positive integer $l$, where
\begin{eqnarray}
\begin{aligned}\label{op-LD-agx}
&\mathbf{\mathcal{U}}=\begin{pmatrix}
           \mathcal{E}\; & \mathcal{E}\;\\
       \mathcal{F}+3\,\partial_y^2 & \mathcal{F}-3\partial_y^2\\
              \end{pmatrix},\qquad \mathcal{F}=-(2P\partial_y+2P_y),\\
&\mathbf{\mathcal{V}}=\frac{1}{54} \begin{pmatrix}
            3\,\partial_y^{-1}\, \Theta \,\partial_y^{-1} & \partial_y^{-1}\, \Theta\,\partial_y^{-1}\left(2P-3\,\partial_y\right)\\
             -3\,\partial_y^{-1}\, \Theta^\ast\,\partial_y^{-1} & -\partial_y^{-1}\, \Theta^\ast\,\partial_y^{-1}\left(2P+3\,\partial_y\right)\\
              \end{pmatrix},
\end{aligned}
\end{eqnarray}
$\mathcal{E}$ and $\Theta$ are defined in \eqref{haop-ga-agx}, while $\mathcal{K}$ and $\mathcal{J}$ are the compatible Hamiltonian operators for the GX system given in  \eqref{haop-gx}.
\end{lemma}

\begin{proof}
The proof relies on an induction argument. First, we prove \eqref{eq-l3.2} in the case of $l=1$. For arbitrary constants $\alpha$ and $\beta$, the direct calculation by use of the transformation \eqref{liou-gx} yields the following operator identities
\begin{eqnarray}
\label{l3.2a}
&&(\partial_y-2\alpha P)\left(\frac{m}{n}\right)^\alpha=\left(\frac{m}{n}\right)^\alpha\,\partial_y, \quad \left(\frac{m}{n}\right)^{\frac{\beta}{2}}(\partial_y+\beta P)\left(\frac{m}{n}\right)^{-\frac{\beta}{2}}=\partial_y.
\end{eqnarray}
Note that \eqref{eq-l3.2} with $l=1$ is equivalent to
\begin{equation*}\label{l3.2b}
\mathbf{B}^{-1}\mathcal{J}=-\mathbf{E}\mathcal{V}\,\mathcal{U}\,\mathbf{E}^{-1}\,\mathbf{B}^{-1}\,\mathcal{K}=\begin{pmatrix}
            \mathscr{R}_{11}\; & \mathscr{R}_{12}\;\\
             \mathscr{R}_{21}\; & \mathscr{R}_{22}\;\\
              \end{pmatrix}.
\end{equation*}
In virtue of the first operator identity in \eqref{gx-op}, we deduce that
\begin{eqnarray*}
\mathbf{B}^{-1}\mathcal{J}=\begin{pmatrix}
            0\;& m^{-1}(\partial_x^2-1)\;\\
             n^{-1}(1-\partial_x^2)&0\\
              \end{pmatrix}= \begin{pmatrix}
           0\; & -\left(\frac{m}{n}\right)^{-\frac{1}{2}}(Q-\partial_y^2)\,\Delta^{\frac{1}{6}}\\
            \left(\frac{m}{n}\right)^{\frac{1}{2}}(Q-\partial_y^2)\,\Delta^{\frac{1}{6}}& 0\;\\
              \end{pmatrix}.
\end{eqnarray*}
On the other hand, a direct calculation shows that
\begin{eqnarray*}
\begin{aligned}
\mathscr{R}_{11}=&\frac{1}{18}\Theta\partial_y^{-1}\Big\{\mathcal{E}\partial_y^{-1}(3\partial_x+m^{-1}m_x+n^{-1}n_x)\Omega^{-1}(3m\partial_x+m_x)\\
&-\left(\partial_y-\frac{2}{3}P\right)\Big[\big(\frac{1}{6}\mathcal{F}\partial_y^{-1}+\frac{1}{2}\partial_y\big)\,\big((\frac{3}{2}\partial_x+m^{-1}m_x)\Omega^{-1}(3m\partial_x+m_x)+\frac{3}{2}\partial_x^{-1}\,m\big)\\
&\quad+\big(\frac{1}{6}\mathcal{F}\partial_y^{-1}-\frac{1}{2}\partial_y\big)\,\big((\frac{3}{2}\partial_x+n^{-1}n_x)\Omega^{-1}(3m\partial_x+m_x)-\frac{3}{2}\partial_x^{-1}\,m\big)\Big]\Big\}\\
=&\frac{1}{2}\Theta\partial_y^{-1}\left(\,n^{-\frac{2}{3}}\partial_x\,m^{\frac{1}{3}}-\Big(\partial_y-\frac{2}{3}P\Big)m^{\frac{2}{3}}n^{-\frac{1}{3}}\right).
\end{aligned}
\end{eqnarray*}
Then using the first equation in \eqref{l3.2a}, we have $\mathscr{R}_{11}=0$. Next,
\begin{eqnarray*}
\begin{aligned}
\mathscr{R}_{12}=&\frac{1}{18}\Theta\partial_y^{-1}\Big\{\mathcal{E}\partial_y^{-1}(3\partial_x+m^{-1}m_x+n^{-1}n_x)\Omega^{-1}(3n\partial_x+n_x)\\
&-\left(\partial_y-\frac{2}{3}P\right)\Big[\big(\frac{1}{6}\mathcal{F}\partial_y^{-1}+\frac{1}{2}\partial_y\big)\,\big((\frac{3}{2}\partial_x+m^{-1}m_x)\Omega^{-1}(3n\partial_x+n_x)-\frac{3}{2}\partial_x^{-1}\,n\big)\\
&\quad+\big(\frac{1}{6}\mathcal{F}\partial_y^{-1}-\frac{1}{2}\partial_y\big)\,\big((\frac{3}{2}\partial_x+n^{-1}n_x)\Omega^{-1}(3n\partial_x+n_x)-\frac{3}{2}\partial_x^{-1}\,n\big)\Big]\Big\}\\
=&\frac{1}{2}\Theta\partial_y^{-1}\left(\Big(\partial_y-\frac{2}{3}P\Big)m^{-\frac{1}{3}}n^{\frac{2}{3}}+m^{-\frac{2}{3}}\partial_x\,n^{\frac{1}{3}}\right)\\
=&\frac{1}{2}\Theta\partial_y^{-1}\left(\;\left(\frac{m}{n}\right)^{-\frac{1}{3}}\,\partial_y\,\left(\frac{m}{n}\right)^{-\frac{1}{6}}+
\left(\frac{m}{n}\right)^{\frac{1}{3}}\,\partial_y\,\left(\frac{m}{n}\right)^{-\frac{5}{6}}\;\right)\Delta^{\frac{1}{6}}\\
=&\Theta\partial_y^{-1}\left(\frac{m}{n}\right)^{-\frac{1}{2}}\left(\partial_y-P\right)\Delta^{\frac{1}{6}}.
\end{aligned}
\end{eqnarray*}

Now we claim that
\begin{equation*}
\Theta\,\partial_y^{-1}\left(\frac{m}{n}\right)^{-\frac{1}{2}}\left(\partial_y-P\right)\left(\frac{m}{n}\right)^{\frac{1}{2}}=-\left(\frac{m}{n}\right)^{-\frac{1}{2}}\left(Q-\partial_y^2\right)\left(\frac{m}{n}\right)^{\frac{1}{2}}.
\end{equation*}
In fact, based on the expression of $Q$ given in \eqref{liou-gx}, a direct calculation leads to
\begin{equation*}
\left(\frac{m}{n}\right)^{-\frac{1}{2}}\left(Q-\partial_y^2\right)\left(\frac{m}{n}\right)^{\frac{1}{2}}=-\Theta,
\end{equation*}
whereas, from the second equation in \eqref{l3.2a},
\begin{equation*}
\Theta\,\partial_y^{-1}\left(\frac{m}{n}\right)^{-\frac{1}{2}}\left(\partial_y-P\right)\left(\frac{m}{n}\right)^{\frac{1}{2}}=\Theta,
\end{equation*}
proving the claim.

After a further calculation, we have
\begin{eqnarray*}
\begin{aligned}
\mathscr{R}_{21}=&-\frac{1}{18}\Theta^\ast\,\partial_y^{-1}\Big\{\mathcal{E}\partial_y^{-1}(3\partial_x+m^{-1}m_x+n^{-1}n_x)\Omega^{-1}(3m\partial_x+m_x)\\
&+\left(\partial_y+\frac{2}{3}P\right)\Big[\big(\frac{1}{6}\mathcal{F}\partial_y^{-1}+\frac{1}{2}\partial_y\big)\,\big((\frac{3}{2}\partial_x+m^{-1}m_x)\Omega^{-1}(3m\partial_x+m_x)+\frac{3}{2}\partial_x^{-1}\,m\big)\\
&\quad+\big(\frac{1}{6}\mathcal{F}\partial_y^{-1}-\frac{1}{2}\partial_y\big)\,\big((\frac{3}{2}\partial_x+n^{-1}n_x)\Omega^{-1}(3m\partial_x+m_x)-\frac{3}{2}\partial_x^{-1}\,m\big)\Big]\Big\}\\
=&-\Theta^\ast\,\partial_y^{-1}\left(\frac{m}{n}\right)^{\frac{1}{2}}\left(\partial_y+P\right)\Delta^{\frac{1}{6}}=\left(\frac{m}{n}\right)^{\frac{1}{2}}(Q-\partial_y^2)\Delta^{\frac{1}{6}}
\end{aligned}
\end{eqnarray*}
and
\begin{eqnarray*}
\begin{aligned}
\mathscr{R}_{22}=&-\frac{1}{18}\Theta^\ast\,\partial_y^{-1}\Big\{\mathcal{E}\partial_y^{-1}(3\partial_x+m^{-1}m_x+n^{-1}n_x)\Omega^{-1}(3n\partial_x+n_x)\\
&+\left(\partial_y+\frac{2}{3}P\right)\Big[\big(\frac{1}{6}\mathcal{F}\partial_y^{-1}+\frac{1}{2}\partial_y\big)\,\big((\frac{3}{2}\partial_x+m^{-1}m_x)\Omega^{-1}(3n\partial_x+n_x)-\frac{3}{2}\partial_x^{-1}\,n\big)\\
&\quad+\big(\frac{1}{6}\mathcal{F}\partial_y^{-1}-\frac{1}{2}\partial_y\big)\,\big((\frac{3}{2}\partial_x+n^{-1}n_x)\Omega^{-1}(3n\partial_x+n_x)+\frac{3}{2}\partial_x^{-1}\,n\big)\Big]\Big\}\\
=&-\frac{1}{2}\Theta^\ast\,\partial_y^{-1}\left[\,m^{-\frac{2}{3}}\partial_x n^{\frac{1}{3}}-\left(\frac{2}{3}P+\partial_y\right)\left(\frac{m}{n}\right)^{-\frac{1}{3}}n^{\frac{1}{3}}\right]=0,
\end{aligned}
\end{eqnarray*}
where we have employed the formulae \eqref{l3.2a} and the relation
\begin{equation*}
\left(\frac{m}{n}\right)^{\frac{1}{2}}\left(Q-\partial_y^2\right)\left(\frac{m}{n}\right)^{-\frac{1}{2}}=-\Theta^\ast.
\end{equation*}

Consequently, combining the preceding results together verifies that  \eqref{eq-l3.2} holds for $l=1$. Then, the induction procedure shows that  \eqref{eq-l3.2} holds in general. Lemma \ref{l3.2} is thus proved.
\end{proof}

\begin{remark}
Notably, we claim that the matrix operators $\mathcal{U}$ and $\mathcal{V}$ defined by \eqref{op-LD-agx} satisfy the following composition identity
\begin{equation}\label{reop-fact-agx}
\overline{\mathcal{R}}=\mathcal{U}\,\mathcal{V}.
\end{equation}
Indeed, define the matrix operator
\begin{equation}\label{reop-fact-agx-1}
\mathbf{\Xi}=\begin{pmatrix}
            0\; & \Theta\;\\
             -\Theta^\ast & 0\;\\
              \end{pmatrix}\mathbf{\Gamma}^\ast.
\end{equation}
It follows from the definition of operators $\mathbf{\Gamma}$ by \eqref{haop-ga-agx} and $\mathcal{U}$ by \eqref{op-LD-agx} that
\begin{equation*}
\mathbf{\Gamma}=-\frac{1}{6}\,\mathcal{U}\,\mathbf{E}^{-1}.
\end{equation*}
Thus, in virtue of \eqref{haop-kj-agx} and \eqref{reop-fact-agx-1}, we have
\begin{equation}\label{reop-fact-agx-2}
\overline{\mathcal{K}}=-\frac{1}{6}\,\mathcal{U}\,\mathbf{E}^{-1}\mathbf{\Xi}.
\end{equation}
On the other hand, the identity
\begin{eqnarray*}
\mathbf{\Xi}\,\overline{\mathcal{J}}^{-1}=\frac{1}{3}\begin{pmatrix}
            0\;& \Theta\;\\
             -\Theta^\ast&0\;\\
              \end{pmatrix}\mathbf{\Gamma}^\ast\begin{pmatrix}
            -6\mathcal{E}^{-1} & 0\;\\
            0\; & -2\partial_y^{-1}\\
              \end{pmatrix}=-\frac{1}{9}\begin{pmatrix}
            3\Theta\partial_y^{-1}& \Theta\partial_y^{-1}\left(2P-3\partial_y\right)\\
           -3\Theta^\ast\partial_y^{-1}&-\Theta^\ast\partial_y^{-1}\left(2P+3\partial_y\right)
              \end{pmatrix},
\end{eqnarray*}
yields
\begin{equation*}\label{reop-fact-agx-3}
\mathbf{E}^{-1}\Xi=-6\mathcal{V}\,\overline{\mathcal{J}}.
\end{equation*}
Substituting the result into \eqref{reop-fact-agx-2}, we arrive at $\overline{\mathcal{R}}=\overline{\mathcal{K}}\,\overline{\mathcal{J}}^{-1}=\mathcal{U}\,\mathcal{V}$, proving the claim.

The formula \eqref{reop-fact-agx} can be viewed as a new operator factorization for the recursion operator $\overline{\mathcal{R}}$, which is not the same as the decomposition of $\overline{\mathcal{R}}=\overline{\mathcal{K}}\,\overline{\mathcal{J}}^{-1}$ using the pair of Hamiltonian operators. We will see from the following proof for Theorem \ref{t3.1} that such a novel factorization plays a key role to identify the systems transformed from the negative (positive) flows of the GX hierarchy to be the corresponding positive (negative) flows of the AGX hierarchy.
\end{remark}

\begin{proof} [\bf{Proof of Theorem \ref{t3.1}}]
A straightforward computation of taking the $t$-derivative of the function $Q(\tau, y)$ and $P(\tau, y)$ according to the Liouville transformation \eqref{liou-gx} yields
\begin{eqnarray}
\begin{aligned}\label{t3.1-eq1}
\begin{pmatrix}
              Q\\
              P
              \end{pmatrix}_\tau=\frac{1}{6}\,\mathcal{U}\mathbf{E}^{-1}\mathbf{B}^{-1}\begin{pmatrix}
              m\\
              n
              \end{pmatrix}_t,
              \end{aligned}
\end{eqnarray}
where $\mathcal{U}$, $\mathbf{E},\mathbf{B}$ are the matrix operators defined in Lemma \ref{l3.2}.

Suppose that $\big(m(t, x),\,n(t, x)\big)$ is the solution of the (GX)$_{-l}$ system \eqref{gx--l'} for some integer $l\geq 1$. Then, using formula \eqref{eq-l3.2} and \eqref{t3.1-eq1}, we find from \eqref{gx--l'} that
\begin{eqnarray*}
\begin{aligned}
\begin{pmatrix}
              Q\\
              P
              \end{pmatrix}_\tau&=\frac{1}{6}\,\mathcal{U}\,\mathbf{E}^{-1}\mathbf{B}^{-1}\LC\mathcal{J} \mathcal{K}^{-1}\RC^{l}\,
              \begin{pmatrix}
              0\\
              0
              \end{pmatrix}=\frac{1}{6}(-1)^l\,\mathcal{U}\LC\mathcal{V}\,\mathcal{U} \RC^{l}\mathbf{E}^{-1}\mathbf{B}^{-1}\,\begin{pmatrix}
              0\\
              0
              \end{pmatrix}\\
              &=\frac{1}{6}(-1)^l\LC\mathcal{U}\,\mathcal{V}\RC^{l}\mathcal{U}\,\begin{pmatrix}
              c\\
              c
              \end{pmatrix}=\frac{2}{3}(-1)^l\,c\,\LC\mathcal{U}\,\mathcal{V}\RC^{l}\,\overline{\mathbf{G}}_1,
              \end{aligned}
\end{eqnarray*}
with $c$ being the corresponding constant of integration. Choosing $c=3(-1)^l/2$,  the fact that $\overline{\mathcal{R}}=\mathcal{U}\mathcal{V}$, immediately reveals that the corresponding  $\big(Q(\tau, y),\,P(\tau, y)\big)$ satisfies the (AGX)$_{l+1}$ system \eqref{agx-l+1}.

Moreover, assume that for each $l\geq 1$, $\big(m(t, x),\,n(t, x)\big)$ solves the (GX)$_{l}$ system \eqref{gx-l}. Then, subject to the transformation \eqref{liou-gx}, the corresponding  $\big(Q(\tau, y),\,P(\tau, y)\big)$ satisfies
\begin{eqnarray*}
\begin{aligned}
\overline{\mathcal{R}}^l\begin{pmatrix}
              Q\\
              P
              \end{pmatrix}_\tau&=\frac{1}{6}\,\overline{\mathcal{R}}^l\mathcal{U}\,\mathbf{E}^{-1}\mathbf{B}^{-1}\LC\mathcal{K}\mathcal{J}^{-1}\RC^{l}\,
              \begin{pmatrix}
              0\\
              0
              \end{pmatrix}=\frac{1}{6}\,\left(\mathcal{U}\,\mathcal{V}\right)^l\,\mathcal{U}\,\mathbf{E}^{-1}\mathbf{B}^{-1}\,\left(\mathcal{K}\mathcal{J}^{-1}\right)^l\begin{pmatrix}
              0\\
              0
              \end{pmatrix}\\
              &=\frac{1}{6}(-1)^l\,\mathcal{U}\,\mathbf{E}^{-1}\,\begin{pmatrix}
              0\\
              0
              \end{pmatrix}=\begin{pmatrix}
              0\\
              0
              \end{pmatrix},
              \end{aligned}
\end{eqnarray*}
where we have made use of Lemma \ref{l3.2} and the factorization $\overline{\mathcal{R}}=\mathcal{U}\,\mathcal{V}$ again. This implies $\big(Q(\tau, y),\,P(\tau, y)\big)$  is a solution for the (AGX)$_{-l}$ system, which completes the proof of the theorem.
\end{proof}

\subsection{The correspondence between the Hamiltonian functionals of the GX and AGX hierarchies}

We now study the correspondence between the Hamiltonian functionals involved in the GX and AGX hierarchies. Based on the pair of Hamiltonian operators $\mathcal{K}$ and $\mathcal{J}$ in \eqref{haop-gx}  for the GX hierarchy and the pair of Hamiltonian operators $\overline{\mathcal{K}}$ and $\overline{\mathcal{J}}$ in \eqref{haop-kj-agx} for the AGX hierarchy, their respective infinite hierarchy of Hamiltonian functionals $\{\mathcal{H}_l\}=\{\mathcal{H}_l(m, n)\}$ and $\{\bar{\mathcal{H}}_l\}=\{\bar{\mathcal{H}}_l(Q, P)\}$ are determined by the following recursive formulae
\begin{equation}\label{gxclhie}
\mathcal{K}\,\delta \mathcal{H}_{l-1}=\mathcal{J}\,\delta \mathcal{H}_l,\quad l\in  \mathbb{Z},
\end{equation}
and
\begin{equation}\label{agxclhie}
\overline{\mathcal{K}}\delta \bar{\mathcal{H}}_{l-1}=\overline{\mathcal{J}} \delta \bar{\mathcal{H}}_{l},\quad l\in \mathbb{Z},
\end{equation}
where $\delta \mathcal{H}_{l}=\LC\delta \mathcal{H}_{l}/\delta m, \, \delta \mathcal{H}_{l}/\delta n \RC^T$ and $\delta \bar{\mathcal{H}}_{l}=\LC\delta \bar{\mathcal{H}}_{l}/\delta Q, \, \delta \bar{\mathcal{H}}_{l}/\delta P \RC^T$. Establishing the correspondence between the two hierarchies of Hamiltonian functionals relies on the following two preliminary lemmas.


\begin{lemma}\label{l3.3}
Let $\{\mathcal{H}_l\}$ and $\{\bar{\mathcal{H}}_l\}$  be the hierarchies of Hamiltonian functionals determined by
\eqref{gxclhie} and \eqref{agxclhie}, respectively. Then, for each $l\in\mathbb{Z}$, their corresponding variational derivatives are related according to the following identity
\begin{equation}\label{eq-l3.3}
\delta \mathcal{H}_{l}(m, n)=6(-1)^{l}\,\mathcal{J}^{-1}\,\mathbf{B}\,\mathbf{\Xi}\,\delta  \bar{\mathcal{H}}_{-(l+1)}(Q, P),
\end{equation}
where the matrix operators $\mathbf{B}$, $\mathbf{\Xi}$ are defined in \eqref{agx-e-b}, \eqref{reop-fact-agx-1} repectively, and $\mathcal{J}$ is the first Hamiltonian operator of the GX hierarchy given by \eqref{haop-gx}.
\end{lemma}

\begin{proof}
First of all, we prove a recursive identity
\begin{equation}\label{l3.3-eq1}
\mathcal{K}^{-1}\,\mathbf{B}\,\mathbf{\Xi}\,\delta  \bar{\mathcal{H}}_{l-1}=-\mathcal{J}^{-1}\,\mathbf{B}\,\mathbf{\Xi}\,\delta  \bar{\mathcal{H}}_{l},\quad l\in \mathbb{Z},
\end{equation}
for the hierarchy of Hamiltonian functionals $\{\bar{\mathcal{H}}_l\}$. In fact, formula \eqref{eq-l3.2} with $l=1$ leads to
\begin{equation*}
\mathcal{K}^{-1}\,\mathbf{B}=-\mathcal{J}^{-1}\,\mathbf{B}\,\mathbf{E}\,\mathcal{V}\,\mathcal{U}\,\mathbf{E}^{-1}.
\end{equation*}
Hence, the left-hand side of equation \eqref{l3.3-eq1} becomes
\begin{equation*}
\mathcal{K}^{-1}\,\mathbf{B}\,\mathbf{\Xi}\,\delta  \bar{\mathcal{H}}_{l-1}=-\mathcal{J}^{-1}\,\mathbf{B}\,\mathbf{E}\,\mathcal{V}\,\mathcal{U}\,\mathbf{E}^{-1}\,\mathbf{\Xi}\,
\overline{\mathcal{K}}^{-1}\,\overline{\mathcal{J}}\,\delta\bar{\mathcal{H}}_{l}.
\end{equation*}
Then, referring back to the definition of operator $\mathbf{\Xi}$ \eqref{reop-fact-agx-1} immediately verifies that \eqref{l3.3-eq1} holds for all $l\in \mathbb{Z}$.

Next, we perform an induction argument to prove \eqref{eq-l3.3}. For the case $l=-2$ of \eqref{eq-l3.3}, due to \eqref{l3.3-eq1} with $l=1$, we need to prove
\begin{equation}\label{l3.3-eq2}
\delta \mathcal{H}_{-2}=-6\,\mathcal{K}^{-1}\,\mathbf{B}\,\mathbf{\Xi}\,\delta  \bar{\mathcal{H}}_{0}.
\end{equation}
In fact, it follows from \eqref{hacl--1-gx} and \eqref{gxclhie} that
\begin{eqnarray}\label{agx-H--2}
\delta \mathcal{H}_{-2}=\mathcal{K}^{-1}\mathcal{J}\delta \mathcal{H}_{-1}=\mathcal{K}^{-1}\mathcal{J}\delta \mathcal{H}_C
=\mathcal{K}^{-1}\begin{pmatrix}
             (\partial_x^2-1)\,m^{\frac{1}{3}}n^{-\frac{2}{3}}\,\\
             (1-\partial_x^2)\,m^{-\frac{2}{3}}n^{\frac{1}{3}}\,
              \end{pmatrix}.
\end{eqnarray}
Using \eqref{liou-gx} and \eqref{gx-op}, we have
\begin{equation*}
(\partial_x^2-1)\,m^{\frac{1}{3}}n^{-\frac{2}{3}}=-\Delta^{\frac{1}{2}}\,(Q-\partial_y^2)\,\left(\frac{m}{n}\right)^\frac{1}{2}=m\,(P_y+P^2-Q)
\end{equation*}
and\begin{equation*}
(1-\partial_x^2)\,m^{-\frac{2}{3}}n^{\frac{1}{3}}=\Delta^{\frac{1}{2}}\,(Q-\partial_y^2)\,\left(\frac{m}{n}\right)^{-\frac{1}{2}}=n\,(P_y-P^2+Q).
\end{equation*}
Thus
\begin{eqnarray}
\begin{aligned}\label{l3.3-eq4}
\delta \mathcal{H}_{-2}&=\mathcal{K}^{-1}\,\mathbf{B}\begin{pmatrix}
            \,P_y+P^2-Q\,\\
             \,P_y-P^2+Q\,
              \end{pmatrix}.
\end{aligned}
\end{eqnarray}
On the other hand, since $\bar{\mathcal{H}}_{0}= \int P\;\mathrm{d}y $, in view of the form of operator $\mathbf{\Xi}$, we deduce that
\begin{eqnarray*}
\begin{aligned}
6\,\mathcal{K}^{-1}\,\mathbf{B}\,\mathbf{\Xi}\,\delta \bar{\mathcal{H}}_{0}&=-\mathcal{K}^{-1}\,\mathbf{B}\,\begin{pmatrix}
            \,\Theta \cdot 1\,\\
             \,-\Theta^\ast \cdot 1\,
              \end{pmatrix},
\end{aligned}
\end{eqnarray*}
which, together with \eqref{agx-H--2} and \eqref{l3.3-eq4}, proves the identity \eqref{l3.3-eq2} and implies that \eqref{eq-l3.3} holds for $l=-2$.

Suppose by induction that  \eqref{eq-l3.3} holds for $l=k$ with $k\leq -2$. Then, we deduce from \eqref{l3.3-eq1} that, for $l=k-1$,
\begin{eqnarray*}
\delta \mathcal{H}_{k-1}=\mathcal{K}^{-1}\mathcal{J}\delta \mathcal{H}_{k}=6(-1)^{k}\,\mathcal{K}^{-1}\,\mathbf{B}\,\mathbf{\Xi}\,\delta  \bar{\mathcal{H}}_{-(k+1)}=6(-1)^{k-1}\,\mathcal{J}^{-1}\,\mathbf{B}\,\mathbf{\Xi}\,\delta  \bar{\mathcal{H}}_{-k},
\end{eqnarray*}
which implies that \eqref{eq-l3.3} holds for each $l \leq -2$.

In the opposite direction, we first prove that
\ben\label{gx-h--1}
\delta \mathcal{H}_{-1}=6\,\mathcal{K}^{-1}\,\mathbf{B}\,\mathbf{\Xi}\,\delta  \bar{\mathcal{H}}_{-1}.
\een
It suffices to show that $\overline{\mathcal{K}}\,\mathbf{\Xi}^{-1}\,\mathbf{B}^{-1}\,\mathcal{K}\delta \mathcal{H}_{-1}=\mathbf{0}$. In fact, by \eqref{reop-fact-agx-2},
\begin{eqnarray*}
\overline{\mathcal{K}}\,\mathbf{\Xi}^{-1}\,\mathbf{B}^{-1}\,\mathcal{K}\delta \mathcal{H}_{-1}=-\frac{1}{6}\,\mathcal{U}\,\mathbf{E}^{-1}\,\mathbf{B}^{-1}\,\mathcal{K}\,\delta \mathcal{H}_{-1}=\frac{1}{6}\,\mathcal{V}^{-1}\,\mathbf{E}^{-1}\,\mathbf{B}^{-1}\,\mathcal{J}\,\delta \mathcal{H}_{-1},
\end{eqnarray*}
where,
\begin{eqnarray*}\begin{aligned}
\mathbf{B}^{-1}\,\mathcal{J}\,\delta \mathcal{H}_{-1}=\begin{pmatrix}
            \,P_y+P^2-Q\,\\
             \,P_y-P^2+Q\,
              \end{pmatrix}=\mathbf{E}\,\mathcal{V}\,\begin{pmatrix}
            \,0\,\\
             \,0\,
              \end{pmatrix},
\end{aligned}
\end{eqnarray*}
so \eqref{gx-h--1} follows.

Further induction on $l$ shows that if \eqref{eq-l3.3} holds for $l=k$ with some $k\geq -1$, then for $l=k+1$,  using \eqref{l3.3-eq1} again, we infer that
\begin{eqnarray*}
\delta \mathcal{H}_{k+1}=\;\mathcal{J}^{-1}\mathcal{K}\,\delta \mathcal{H}_{k}=6(-1)^{k}\,\mathcal{J}^{-1}\,\mathcal{K}\,\mathbf{B}\,\mathbf{\Xi}\,\delta  \bar{\mathcal{H}}_{-(k+1)}=6(-1)^{k+1}\,\mathcal{J}^{-1}\,\mathbf{B}\,\mathbf{\Xi}\,\delta  \bar{\mathcal{H}}_{-(k+2)},
\end{eqnarray*}
which completes the induction step, and proves \eqref{eq-l3.3}  in general.
\end {proof}

The next lemma reveals the effect of the Liouville transformation \eqref{liou-gx} on the variational derivatives, which follows from Theorem \ref{t2}.

\begin{lemma}\label{l3.4}
Let $\big(m(t, x),\,n(t, x)\big)$ and $\big(Q(\tau, y),\,P(\tau, y)\big)$ be related by the Liouville transformation \eqref{liou-gx}. If
$\mathcal{H}(m, n)=\bar{\mathcal{H}}(Q, P)$, then
\begin{equation*}\label{eq-l3.4}
\delta \mathcal{H}(m, n)=-\frac{1}{6}\,\Delta^{\frac{1}{3}}\,\mathbf{B}^{-1}\,\mathbf{E}^{-1}\,\mathcal{U}^\ast\,\delta \bar{\mathcal{H}}(Q, P).
\end{equation*}
\end{lemma}

Finally, we claim that
\begin{eqnarray*}
\begin{aligned}
\mathcal{J}^{-1}\,\mathbf{B}\,\begin{pmatrix}
              0\,&\Theta\\
             -\Theta^\ast & 0
              \end{pmatrix}=\Delta^{\frac{1}{3}}\,\mathbf{B}^{-1},
              \end{aligned}
\end{eqnarray*}
which is equivalent to
\begin{eqnarray*}
\begin{aligned}
\mathbf{B}\,\begin{pmatrix}
              0\,&\Theta\\
             -\Theta^\ast & 0
              \end{pmatrix}=\mathcal{J}\,\begin{pmatrix}
              m^{-\frac{2}{3}}n^{\frac{1}{3}}&0\,\\
            0\, & m^{\frac{1}{3}}n^{-\frac{2}{3}}
              \end{pmatrix}.
              \end{aligned}
\end{eqnarray*}
In fact,
\begin{eqnarray*}
\mathcal{J}\,\begin{pmatrix}
              m^{-\frac{2}{3}}n^{\frac{1}{3}}&0\,\\
            0\, & m^{\frac{1}{3}}n^{-\frac{2}{3}}
              \end{pmatrix}=\begin{pmatrix}
             0&\Delta^{\frac{1}{2}}(Q-\partial_y^2)\left(\frac{m}{n}\right)^{\frac{1}{2}}\,\\
            \Delta^{\frac{1}{2}}(Q-\partial_y^2)\left(\frac{m}{n}\right)^{-\frac{1}{2}} & 0\;
              \end{pmatrix}
              =\mathbf{B}\,\begin{pmatrix}
              0\,&\Theta\\
             -\Theta^\ast & 0
              \end{pmatrix}.
\end{eqnarray*}

Now, define the functional
\begin{equation*}
\mathcal{G}_{l}(Q, P)\equiv \mathcal{H}_{l}(m, n),
\end{equation*}
we have, on the one hand, by Lemma \ref{l3.4}
\begin{equation*}\label{t3.3-1}
\delta \mathcal{H}_{l}(m, n)=-\frac{1}{6}\,\Delta^{\frac{1}{3}}\,\mathbf{B}^{-1}\,\mathbf{E}^{-1}\,\mathcal{U}^\ast\,\delta \mathcal{G}_{l}(Q,  P).
\end{equation*}
On the other hand, in view of Lemma \ref{l3.3},
\begin{eqnarray*}
\begin{aligned}
\delta \mathcal{H}_{l}(m, n)&=6(-1)^{l}\,\mathcal{J}^{-1}\,\mathbf{B}\,\mathbf{\Xi}\,\delta  \bar{\mathcal{H}}_{-(l+1)}(Q, P)\\
&=(-1)^{l}\,\mathcal{J}^{-1}\,\mathbf{B}\,\begin{pmatrix}
              0\,&\Theta\\
             -\Theta^\ast & 0
              \end{pmatrix}\,\mathbf{E}^{-1}\,\mathcal{U}^\ast\,\delta \bar{\mathcal{H}}_{-(l+1)}(Q, P)\\
   &= (-1)^{l} \Delta^{\frac{1}{3}}\,\mathbf{B}^{-1}\,\mathbf{E}^{-1}\,\mathcal{U}^\ast\,\delta \bar{\mathcal{H}}_{-(l+1)}(Q, P).  \end{aligned}
\end{eqnarray*}
Combining the preceding two equations gives rise to
\begin{equation*} \mathcal{H}_l(m, n)=\mathcal{G}_{l}(Q, P)=6(-1)^{l+1}\,\bar{\mathcal{H}}_{-(l+1)}(Q, P),\end{equation*}
which establishes the correspondence between the sequences of the Hamiltonian quantities admitted by the GX and AGX hierarchies. Thus, we have proved the following theorem.

\begin{theorem}\label{t3.3}
For any nonzero integer $l$, each Hamiltonian conserved density  $\mathcal{H}_l(m, n)$ of the GX hierarchy relates the Hamiltonian conserved density $\bar{\mathcal{H}}_{l}(Q, P)$, under the Liouville transformation \eqref{liou-gx},  according to the following identity
\begin{equation*}\label{cl-rela-dp3}
\mathcal{H}_l(m, n)=6(-1)^{l+1}\bar{\mathcal{H}}_{-(l+1)}(Q, P), \quad 0\neq l\in \mathbb{Z}.
\end{equation*}
\end{theorem}
\section{The Liouville correspondence for the dual DWW hierarchy}

The following dispersive water wave (DWW) system
\begin{eqnarray}\label{dww}
\qxeq{ q_t= \big( -q_x+2q r\big)_x,\\
r_t=\big( r_x+r^2+2q\big)_x,}
\end{eqnarray}
is an integrable physical system describing the propagation of shallow water waves \cite{kau, kup85, whi}. The bi-Hamiltonian formulation of system \eqref{dww} supports the tri-Hamiltonian dual structure and its dual counterpart is recognized as the dual dispersive water wave (dDWW) system, proposed in \cite{kloq3}, of the form
\ben\label{dDWW}
\eeq{\rho_t=\LC (\rho+v)\,u\RC_x,&\rho=v-v_x,\\
\gamma_t=\LC \gamma \,u+2v\RC_x,& \gamma=u+u_x.}
\een
The bi-Hamiltonian structure for the dDWW system \eqref{dDWW}
\begin{eqnarray*}
\begin{aligned}
\begin{pmatrix}
              \rho\\
              \gamma
              \end{pmatrix}_t=\mathcal{K}\,\delta \mathcal{H}_1=\mathcal{J}\,\delta \mathcal{H}_2,\qquad \delta \mathcal{H}_n=\left(\,\frac{\delta \mathcal{H}_n}{\delta \rho},\;\frac{\delta \mathcal{H}_n}{\delta \gamma}\,\right)^T,\quad n=1, 2,
\end{aligned}\end{eqnarray*}
is governed by the following pair of Hamiltonian operators \cite{kloq3}
\begin{eqnarray}
\begin{aligned}\label{haop-dDWW}
\mathcal{K}=\begin{pmatrix}
            \rho \partial_x+\partial_x \rho & \gamma \partial_x\\
             \partial_x\,\gamma & 2\partial_x\\
              \end{pmatrix},\qquad
\mathcal{J}=\begin{pmatrix}
        0 & \partial_x-\partial_x^2\\
             \partial_x+\partial_x^2& 0\\
              \end{pmatrix},
\end{aligned}
\end{eqnarray}
together with the associated Hamiltonian functionals
\begin{equation*}
\qxeq{\mathcal{H}_1=\mathcal{H}_1(\rho, \gamma)= \int  (u-u_x)\,v\;\mathrm{d}x \qquad \hbox{and} \qquad  \mathcal{H}_2=\mathcal{H}_2(\rho, \gamma)=\int  (u+u_x) \,u v\;\mathrm{d}x.}
\end{equation*}
The members in the positive direction of the dDWW integrable hierarchy
\begin{eqnarray}
\begin{aligned}\label{dDWW-n}
\begin{pmatrix}
              \rho\\
              \gamma
              \end{pmatrix}_t=\mathbf{M}_n=\mathcal{K}\,\delta \mathcal{H}_{n-1}(\rho, \gamma)=\mathcal{J}\,\delta \mathcal{H}_n(\rho, \gamma),\quad n=1, 2, \ldots,
\end{aligned}
\end{eqnarray}
are obtained by applying successively the recursion operator $\mathcal{R}=\mathcal{K}\,\mathcal{J}^{-1}$ to the seed system
\begin{eqnarray*}
\begin{aligned}
\begin{pmatrix}
              \rho\\
              \gamma
              \end{pmatrix}_t=\mathbf{M}_1=\begin{pmatrix}
              \rho\\
              \gamma
              \end{pmatrix}_x=\mathcal{J}\,\delta \mathcal{H}_1(\rho, \gamma).
\end{aligned}
\end{eqnarray*}
Observe that the Hamiltonian operator $\mathcal{K}$ in \eqref{haop-dDWW} admits the following Casimir functional
\begin{equation*}
\qxeq{\mathcal{H}_C(\rho, \gamma)=-\frac{1}{2} \int  \sqrt{4\rho-\gamma^2}\;\mathrm{d}x,}
\end{equation*}
with variational derivative
\begin{equation*}
\delta \mathcal{H}_C=\left(\,\frac{\delta \mathcal{H}_C}{\delta \rho},\;\frac{\delta \mathcal{H}_C}{\delta \gamma}\,\right)^T=\left(-\frac{1}{\sqrt{4\rho-\gamma^2}},\; \frac{\gamma}{2\sqrt{4\rho-\gamma^2}}\right)^T.
\end{equation*}
The functional $\mathcal{H}_C$ leads to the associated Casimir system
\begin{eqnarray*}
\begin{aligned}
\begin{pmatrix}
              \rho\\
              \gamma
              \end{pmatrix}_t=\mathbf{M}_{-1}=\mathcal{J}\,\delta \mathcal{H}_C,
\end{aligned}
\end{eqnarray*}
which is given explicitly by
\begin{equation}\label{dDWW--1}
 \rho_t=\frac{1}{2}(\partial_x-\partial_x^2)\frac{\gamma}{\sqrt{4\rho-\gamma^2}}, \quad \gamma_t=-(\partial_x+\partial_x^2)\frac{1}{\sqrt{4\rho-\gamma^2}},
\end{equation}
and serves as the first negative flow in the dDWW integrable hierarchy. Starting from the Casimir system \eqref{dDWW--1}, one (formally) can construct an infinite hierarchy of higher-order commuting bi-Hamiltonian systems and corresoponding Hamiltonian functionals $\{\mathcal{H}_{-n}\}=\{\mathcal{H}_{-n}(\rho, \gamma)\}$ in the negative direction:
\begin{eqnarray}
\begin{aligned}\label{dDWW--n}
\begin{pmatrix}
              \rho\\
              \gamma
              \end{pmatrix}_t=\mathbf{M}_{-n}=\mathcal{K}\,\delta \mathcal{H}_{-(n+1)}(\rho, \gamma)=\mathcal{J}\,\delta \mathcal{H}_{-n}(\rho, \gamma),\quad n=1, 2, \ldots,
\end{aligned}
\end{eqnarray}
with $\mathcal{H}_{-1}(\rho, \gamma)=\mathcal{H}_C(\rho, \gamma)$ and $\delta \mathcal{H}_{-n}(\rho, \gamma)=\big(\,\delta \mathcal{H}_{-n}/\delta \rho,\;\delta \mathcal{H}_{-n}/\delta \gamma\,\big)^T$.

Now, we define $\omega(\rho, \gamma)=\sqrt{4\rho-\gamma^2}$ and introduce the following coordinate transformation:
\begin{eqnarray}\label{liou-dDWW}
\eeq{y=\int^x\!\omega(t,\,\xi)\,\mathrm{d}\xi, & \tau=t,\\
Q(\tau, y)=\frac{1}{\omega(t, x)}\,\LC 1-\frac{\omega_x(t, x)}{\omega(t, x)}\RC, &  P(\tau, y)=-\frac{\gamma(t, x)}{\omega(t, x)}.}
\end{eqnarray}
In analogy with the 2CH and GX hierarchies, in this section we aim to investigate how the transformation \eqref{liou-dDWW} affect the underlying correspondence between the flows of the dDWW hierarchy and what we will call the {\it associated dDWW} (AdDWW) integrable hierarchy.

We first investigate the integrable structure of the AdDWW hierarchy. As a direct application of Theorem \ref{t2}, one can readily construct its pair of Hamiltonian operators, from the known pair of Hamiltonian operators $\mathcal{K}$ and $\mathcal{J}$ given by \eqref{haop-dDWW} of the dDWW system, through the transformation \eqref{liou-dDWW}.

\begin{theorem}\label{t4.1}
Under the coordinate transformation  \eqref{liou-dDWW}, the Hamiltonian pair $\mathcal{K}$ and $\mathcal{J}$ \eqref{haop-dDWW} admitted by the dDWW system \eqref{dDWW} is converted to a pair of  Hamiltonian operators
\begin{eqnarray}
\begin{aligned}\label{haop-AdDWW}
&\overline{\mathcal{K}}=2\begin{pmatrix}
            \mathcal{X}\big(P\mathcal{X}+\mathcal{O}\,P\,\partial_y^{-1}\,\big)\,\mathcal{O} & \mathcal{X}\big(P\mathcal{X}\,P\partial_y+\mathcal{O}\,(1+P\partial_y^{-1}P\partial_y)\big)\\
          \mathcal{Y}\,\mathcal{X}\,\mathcal{O}+(P_y\partial_y^{-1}+P)\,\mathcal{O}\,P\,\partial_y^{-1}\,\mathcal{O} & \mathcal{Y}\,\mathcal{X}\,P\partial_y+(P_y\partial_y^{-1}+P)\,\mathcal{O}\,(1+P\partial_y^{-1}P\partial_y)\\
              \end{pmatrix},\\
&\overline{\mathcal{J}}=2\begin{pmatrix}
  \mathcal{X}\, \mathcal{O} & \chi\,P\,\partial_y\,\\
            \partial_yP\partial_y^{-1}\, \mathcal{O} & \partial_y(1+P\partial_y^{-1}P\partial_y)\\
              \end{pmatrix},
\end{aligned}
\end{eqnarray}
where the operators $\mathcal{O}$, $\mathcal{X}$ and $\mathcal{Y}$ are defined by
\begin{equation}
\label{OXY}
\mathcal{O}=Q\partial_y-\partial_y^2,\quad \mathcal{X}=\partial_y+Q+Q_y\partial_y^{-1}, \quad  \mathcal{Y}=1+P^2+P_y\partial_y^{-1}P.
\end{equation}
\end{theorem}

\begin{proof}
Let $b$, $\nu$ be arbitrary nonzero constants.  In a similar manner as Lemma \ref{l2.1}, we can verify the following operator identities
\begin{eqnarray}
\label{t4.1e}
\eeq{\omega^{-(1+\nu)}\,\left(1+b\,\partial_x^2\right)\,\omega^{\nu}=Q+b\,\partial_y,\\
\omega^{-2}\,(\partial_x-\partial_x^2)=Q\,\partial_y-\partial_y^2\,,\\
\omega^{-1}\,(\partial_x+\partial_x^2)\,\omega^{-1}=\partial_y Q\,+\partial_y^2\,,\\
\omega^{-2}\,(\rho\,\partial_x+\partial_x\,\rho)\,\omega^{-1}=\frac{1}{2}(\partial_y+P\,\partial_y\,P).}
\end{eqnarray}
The first and second identities in \eqref{t4.1e} imply that the matrix operators $\mathbf{T}$ and $\mathbf{T}^\ast$ in Theorem \ref{t2} are
\begin{eqnarray}
\begin{aligned}\label{dww-T}
\mathbf{T}=-\omega \begin{pmatrix}
            2\mathcal{X}\,\omega^{-2} & \mathcal{X}\,P\,\omega^{-1}\\
           2(P+P_y\partial_y^{-1})\omega^{-2} & \mathcal{Y}\,\omega^{-1}\\
              \end{pmatrix},\quad
             \mathbf{T}^\ast=-\begin{pmatrix}
 2\omega^{-1}\,\partial_y^{-1}\,\mathcal{O} & 2\omega^{-1}\,\partial_y^{-1}P\partial_y\\
            P\,\partial_y^{-1}\,\mathcal{O} & 1+P\partial_y^{-1}P\partial_y\\
              \end{pmatrix}.
\end{aligned}
\end{eqnarray}
Therefore, in view of  \eqref{t4.1e}, the formulas for the Hamiltonian operators \eqref{haop-AdDWW} follow directly from formula \eqref{rela-haop-n}.
\end{proof}

Now, with the Hamiltonian pair $\overline{\mathcal{K}}$ and $\overline{\mathcal{J}}$ given by \eqref{haop-AdDWW} in hand, the positive flows in the AdDWW hierarchy are generated by applying the recursion operator
\begin{eqnarray}\label{reop-AdDWW}
\overline{\mathcal{R}}=\overline{\mathcal{K}}\,\overline{\mathcal{J}}^{-1}=\begin{pmatrix}
 \mathcal{X}\,P & \mathcal{X}\,(Q-\partial_y)\\
           \mathcal{Y} & (P+P_y\,\partial_y^{-1})(Q-\partial_y)\\
              \end{pmatrix}
\end{eqnarray}
to the seed symmetry $\overline{\mathbf{M}}_1=(Q_y,\,P_y)^T$ successively:
\begin{equation}\label{AdDWW-n+1}
\begin{pmatrix}
             Q\\
             P
              \end{pmatrix}_\tau=\overline{\mathbf{M}}_n=\overline{\mathcal{R}}^{n-1}\overline{\mathbf{M}}_1,\quad n=1, 2, \ldots,
\end{equation}
whereas the negative flows are
\begin{equation}\label{AdDWW--n}
\overline{\mathcal{R}}^n\begin{pmatrix}
             Q\\
             P
              \end{pmatrix}_\tau=\overline{\mathbf{M}}_0=\begin{pmatrix}
             0\\
             0
              \end{pmatrix}, \quad n=1, 2, \ldots.
\end{equation}


\begin{lemma}\label{l4.1}  Let  $\mathcal{K}$,  $\mathcal{J}$, be the compatible Hamiltonian operators \eqref{haop-dDWW} for the dDWW hierarchy and $\overline{\mathcal{R}}$ the recursion operator \eqref{reop-AdDWW} admitted by the AdDWW hierarchy.  Let
\begin{eqnarray*}
\begin{aligned}
\mathbf{C}=\frac{1}{2}\begin{pmatrix}
          0\; & -\omega^2(Q-\partial_y)\\
       -2\omega & 0\\
              \end{pmatrix}.
\end{aligned}
\end{eqnarray*}
Then, subject to the transformation \eqref{liou-dDWW}, we have the operator identity
\begin{equation}\label{eq-l4.1}
\mathbf{C}^{-1}\left( \mathcal{J}\,\mathcal{K}^{-1}\right)^n \mathbf{C}= \overline{\mathcal{R}}^n \qquad \hbox{\it for all} \qquad 0 < n \in \Z.
\end{equation}

\end{lemma}

Employing this Lemma and an induction procedure, we are able to obtain the following result on the correspondence between the dDWW and AdDWW hierarchies. Hereafter, we denote the $n$-th system in the positive and negative directions of the dDWW hierarchy by  (dDWW)$_n$ and (dDWW)$_{-n}$, respectively, while the $n$-th positive and negative flows in the AdDWW hierarchy by (AdDWW)$_{n}$ and (AdDWW)$_{-n}$, respectively.

\begin{theorem}\label{t4.2}
Under the transformation \eqref{liou-dDWW}, for each integer $n \in  \mathbb{Z}$, (dDWW)$_n$ system is related to the  (AdDWW)$_{-(n-1)}$ system.
\end{theorem}

The proof of this theorem is based mainly on the operator decomposition identity \eqref{eq-l4.1} and is omitted for brevity.


We now in a position to establish  the correspondence between the Hamiltonian functionals of the dDWW and AdDWW hierarchies. In particular, for the dDWW hierarchy, the bi-Hamiltonian structure \eqref{dDWW-n},  \eqref{dDWW--n} produce the bi-infinite sequence of functionals $\{\mathcal{H}_n\}$ by \begin{equation}\label{cl-hie-dDWW}
\mathcal{K}\,\delta \mathcal{H}_{n}=\mathcal{J}\,\delta \mathcal{H}_{n+1},\qquad n\in  \mathbb{Z}.
\end{equation}
On the other hand, the recursive formula
\begin{equation}\label{cl-hie-AdDWW}
\overline{\mathcal{K}}\delta \bar{\mathcal{H}}_{n}=\overline{\mathcal{J}} \delta \bar{\mathcal{H}}_{n+1},\qquad n\in \mathbb{Z},
\end{equation}
with the Hamiltonian pair \eqref{haop-AdDWW} of the AdDWW hierarchy, gives rise to an infinite sequence of Hamiltonian functionals $\{\bar{\mathcal{H}}_n\}$ admitted by the AdDWW flows \eqref{AdDWW-n+1} and \eqref{AdDWW--n}.

The formula for the correspondence between the variational derivatives $\delta \mathcal{H}_{n}(\rho, \gamma)$ and $\delta  \bar{\mathcal{H}}_{n}(Q, P)$ can be proved by a straightforward induction.

\begin{lemma}\label{l4.2}
Let $\{\mathcal{H}_n\}$ and $\{\bar{\mathcal{H}}_n\}$  be the hierarchies of Hamiltonian functionals determined by the recursive formulae \eqref{cl-hie-dDWW} and \eqref{cl-hie-AdDWW}, respectively. Then, for each $n\in\mathbb{Z}$, their respective variational derivatives  satisfy the following identiy
\begin{equation*}\label{eq-l4.2}
\delta \mathcal{H}_{n}(\rho, \gamma)=\mathcal{K}^{-1}\,\mathbf{C}\,\bar{\mathcal{J}}\,\delta  \bar{\mathcal{H}}_{-(n+1)}(Q, P).
\end{equation*}
\end{lemma}

In addition, a formula for the change of the variational derivative under  the transformation \eqref{liou-dDWW} is given in the following lemma, which is also a direct consequence of Lemma \ref{l2.3}.

\begin{lemma}\label{l4.3}
Let $\big(\rho(t, x), \gamma(t, x)\big)$ and $\big(Q(\tau, y),\,P(\tau, y)\big)$ be related by the Liouville transformation \eqref{liou-dDWW}. If $\mathcal{H}(\rho, \gamma)=\bar{\mathcal{H}}(Q, P)$, then
\begin{equation}\label{eq-l4.3}
\delta \mathcal{H}(\rho, \gamma)=\mathbf{T}^\ast\,\delta \bar{\mathcal{H}}(Q, P),
\end{equation}
where $\mathbf{T}^\ast$ is the formal adjoint of operators $\mathbf{T}$ given in \eqref{dww-T}.
\end{lemma}

Finally, referring back to the form of the Hamiltonian operators $\overline{\mathcal{J}}$, using the identity
\begin{eqnarray*}
\bar{\mathbb R}=\Delta^{-1}{\bf T}{\bf C},
\end{eqnarray*}
and \eqref{eq-l4.1} with $n=1$, one has
\begin{equation}\label{t4.3-eq1}
\mathbf{T}^\ast=\mathcal{J}^{-1}\,\mathbf{C}\,\overline{\mathcal{J}}.
\end{equation}
It follows that
\begin{eqnarray*}
\delta \mathcal{H}_{n}(\rho, \gamma)=\mathcal{K}^{-1}\,\mathbf{C}\,\overline{\mathcal{J}}\,\delta  \bar{\mathcal{H}}_{-(n+1)}(Q, P)=\mathcal{J}^{-1}\,\mathbf{C}\,\overline{\mathcal{J}}\,\delta  \bar{\mathcal{H}}_{-n}(Q, P)=\mathbf{T}^\ast\,\delta  \bar{\mathcal{H}}_{-n}(Q, P).
\end{eqnarray*}
Hence, based on the hypothesis of Lemma  \ref{l4.3}, we define the functional
\begin{equation*}
\mathcal{G}_{n}(Q, P)\equiv \mathcal{H}_{n}(\rho, \gamma).
\end{equation*}
Then, together with \eqref{eq-l4.3} and \eqref{t4.3-eq1}, we have
\begin{equation*}
\mathbf{T}^\ast\,\delta \mathcal{G}_{n}(Q, P)=\delta \mathcal{H}_{n}(\rho, \gamma)=\mathbf{T}^\ast\,\delta  \bar{\mathcal{H}}_{-n}(Q, P),
\end{equation*}
which yields
\begin{equation*}
\delta \mathcal{G}_{n}(Q, P)=\delta  \bar{\mathcal{H}}_{-n}(Q, P).
\end{equation*}
Then
\begin{equation*}
\mathcal{H}_{n}(\rho, \gamma)=\bar{\mathcal{H}}_{-n}(Q, P).
\end{equation*}
Consequently, we have now proved the following main theorem on the Hamiltonian functionals of two hierarchies.

\begin{theorem}\label{t4.3}
Each Hamiltonian functional $\mathcal{H}_n(\rho, \gamma)$ of the dDWW hierarchy yields a Hamiltonian functionals of the AdDWW hierarchy, under the transformation \eqref{liou-dDWW},  according to the following identity
\begin{equation*}
\bar{\mathcal{H}}_{-n}(Q, P)=\mathcal{H}_{n}(\rho, \gamma), \qquad n\in \mathbb{Z}.
\end{equation*}
\end{theorem}

\noindent {\bf Acknowledgements.}
 Kang's research was supported by NSFC under Grant 11631007 and Grant 11471260. Liu's research was supported in part by NSFC under Grant 11722111 and Grant 11631007.  Qu's research was supported by NSFC under Grant 11631007 and Grant 11471174.

\end{document}